\documentclass[sigconf,nonacm]{aamas}

\usepackage{hyperref}
\usepackage{balance}

\setcopyright{ifaamas}
\acmConference[AAMAS '26]{Proc.\@ of the 25th International Conference
on Autonomous Agents and Multiagent Systems (AAMAS 2026)}{May 25 -- 29, 2026}
{Paphos, Cyprus}{C.~Amato, L.~Dennis, V.~Mascardi, J.~Thangarajah (eds.)}
\copyrightyear{2026}
\acmYear{2026}
\acmDOI{}
\acmPrice{}
\acmISBN{}

\acmSubmissionID{1228}

\title[]{Metric Hedonic Games on the Line}

\author{Merlin de la Haye}
\affiliation{
  \institution{Hasso Plattner Institute}
  \city{Potsdam}
  \country{Germany}}
\email{Merlin.delaHaye@hpi.de}

\author{Pascal Lenzner}
\affiliation{
  \institution{University of Augsburg}
  \city{Augsburg}
  \country{Germany}}
\email{pascal.lenzner@uni-a.de}

\author{Farehe Soheil}
\affiliation{
  \institution{Hasso Plattner Institute}
  \city{Potsdam}
  \country{Germany}}
\email{Farehe.Soheil@hpi.de}

\author{Marcus Wunderlich}
\affiliation{
  \institution{University of Augsburg}
  \city{Augsburg}
  \country{Germany}}
\email{marcus.wunderlich@uni-a.de}

\begin{abstract}
Hedonic games are fundamental models for investigating the formation of coalitions among a set of strategic agents, where every agent has a certain utility for every possible coalition of agents it can be part of. To avoid the intractability of defining exponentially many utilities for all possible coalitions, many variants with succinct representations of the agents' utility functions have been devised and analyzed, e.g., modified fractional hedonic games by Monaco et al.~\cite{MonacoStableOutcomesModified2020}. 
We extend this by studying a novel succinct variant that is related to modified fractional hedonic games. In our model, each agent has a fixed type-value and an agent's cost for some given coalition is based on the differences between its value and those of the other members of its coalition. This allows to model natural situations like athletes forming training groups with similar performance levels or voters that partition themselves along a political spectrum. 

In particular, we investigate natural variants where an agent's cost is defined by distance thresholds, or by the maximum or average value difference to the other agents in its coalition. For these settings, we study the existence of stable coalition structures, their properties, and their quality in terms of the price of anarchy and the price of stability. Further, we investigate the impact of limiting the maximum number of coalitions. Despite the simple setting with metric distances on a line, we uncover a rich landscape of models, partially with counter-intuitive behavior. Also, our focus on both swap stability and jump stability allows us to study the influence of fixing the number and the size of the coalitions. Overall, we find that stable coalition structures always exist but that their properties and quality can vary widely. 
\end{abstract}


\newcommand{\BibTeX}{\rm B\kern-.05em{\sc i\kern-.025em b}\kern-.08em\TeX}

\usepackage{mathtools}
\usepackage{csquotes}
\usepackage{multirow,multicol}
\usepackage{nicematrix}
\usepackage[normalem]{ulem}
\usepackage{thm-restate}
\usepackage{thmtools}
\usepackage{cleveref}
\usepackage{graphicx}
\usepackage{todonotes}

\usepackage{xcolor}

\newcommand{\N}{\mathbb{N}}
\newcommand{\R}{\mathbb{R}}
\newcommand{\nh}{\lfloor n/2 \rfloor}
\newcommand{\nhc}{\lceil n/2 \rceil}

\newcommand{\soc}[1]{\text{SC}(#1)}

\DeclareMathOperator{\dist}{d}
\DeclareMathOperator{\cost}{cost}
\DeclareMathOperator{\SC}{SC}
\DeclareMathOperator{\OPT}{OPT}
\DeclareMathOperator{\EQ}{EQ}
\DeclareMathOperator{\PoA}{PoA}
\DeclareMathOperator{\PoS}{PoS}
\DeclareMathOperator{\SW}{SW}
\DeclareMathOperator{\avg}{avg}
\DeclareMathOperator{\bigO}{\mathcal{O}}

\newcommand{\cev}[1]{\reflectbox{\ensuremath{\vec{\reflectbox{\ensuremath{#1}}}}}}

\Crefname{prop}{Proposition}{Propositions}

\Crefname{obs}{Observation}{Observations}

\newcommand{\doublecheckmark}{\checkmark\hspace{-5pt}\checkmark}
\newcommand{\existsnotall}{\exists\hspace{-1pt}{\not\hspace{-1pt}\forall}}
\begin{document}

\pagestyle{fancy}
\fancyhead{}

\maketitle

\section{Introduction}
Consider athletes preparing for a marathon who want to form running groups that are supervised by coaches.
Each athlete has an individual performance level (e.g. their pace) and wants to be in a group with other athletes who are on a similar level.
The utility of an athlete then only depends on the comparison between their individual performance level and those of the other athletes in their running group.
As there is only a limited number of coaches available, the maximum number of running groups is restricted.

This setting belongs to the research field of agent-based coalition formation that has been widely studied in the last decades, starting with the work of~\citet{MillerCoalitionFormationCharacteristic1980} on coalition formation games and later the seminal concept of hedonic games studied by~\citet{DrezeHedonicCoalitionsOptimality1980,BogomolnaiaStabilityHedonicCoalition2002} and~\citet{BanerjeeCoreSimpleCoalition2001}; see also the surveys by~\citet{AzizHedonicGames2016} and~\citet{HajdukovaCoalitionFormationGames2006}.

The key feature of hedonic games is that an agent's utility depends only on the members of its coalition and not on external factors such as the composition or size of other coalitions. 
In our athletics example, the training quality of an individual athlete only depends on the other athletes running in the same group with them.

In their full generality, hedonic games are hard to describe and analyze due to the existence of exponentially many possible coalitions and valuations. However, many natural coalition formation settings can be captured by hedonic games that have a succinct representation.
Classical examples are additively separable hedonic games~\cite{BogomolnaiaStabilityHedonicCoalition2002}, and fractional hedonic games~\cite{AzizFractionalHedonicGames2019, MonacoStableOutcomesModified2020}, where agents value each other individually and their utility in a coalition is the sum or average over all the valuations of agents in their~coalition.

Numerous succinct variants have been proposed; see the survey by~\citet{KerkmannHedonicGamesOrdinal2020}. However, none of them closely captures our setting, in which agents have fixed real-valued types and derive their utility from the distances between these types.

In our model, instead of considering utilities, we take the equivalent cost-based approach and 
focus on three natural cost functions that depend only on the type differences within a coalition: 
\begin{itemize}
\item \textsc{Average}: the cost of an agent is the average type difference to the other agents (without self-effect),
\item \textsc{Maximum}: the cost of an agent is the maximum type difference to the other agents,
\item \textsc{Cutoff}: the cost of an agent is the number of other agents who have a type difference larger than a given threshold $\lambda$.
\end{itemize}
See \Cref{exmp:intro} for an illustration of the cost functions. These cost functions serve as a proof of concept for capturing natural agent behavior. Moreover, they have been recently used in a closely related agent-based model investigating residential segregation~\cite{BiloSchellingGamesContinuous2023}. 

\begin{figure}[t]
    \centering
    \includegraphics[width=0.9\linewidth,page=1]{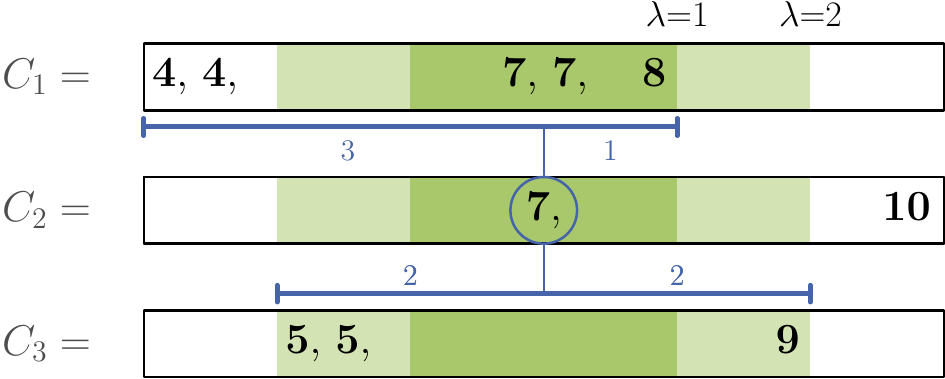}
    \caption{Coalition structure $\mathcal{C}=\{C_1,C_2,C_3\}$ with $n=10$ agents and $k=3$ coalitions. Bold black numbers represent agents with this value. Green areas are the intervals with radius $\lambda=1$ (green) and $\lambda=2$ (light green) around the value $7$. The blue lines show the maximum distance of the circled agent to both sides for coalitions $C_1$ and~$C_3$.}
    \label{fig:example}
    \Description{Three coalitions are shown as three rows, where as the utilities are as follows: $C_1=\{4,4,7,7,8\}$, $C_2=\{7,10\}$, $C_3=\{5,5,9\}$. Agent 7 of coalition $C_2$ is highlighted. Example 1.1 discusses its utilities for the different cost functions.}
\end{figure}

\begin{example}\label{exmp:intro}
    Consider the game instance with coalition structure, i.e., a partition of the agents into disjoint coalitions, $\mathcal{C}=\{C_1,C_2,C_3\}$, shown in \Cref{fig:example}. In this example the circled agent with value $7$, called agent $x$, is currently in its worst coalition~$C_2$ (for all three cost functions) but would prefer to join either coalition $C_1$ or $C_3$, depending on the cost function (and on $\lambda$ for \textsc{Cutoff}). With the \textsc{Average} cost function, the agents in coalitions $C_1$ and $C_3$ have an average distance of $1.4$ and $2$ to agent $x$, respectively. Thus, agent $x$ prefers $C_1$ over $C_3$ in this setting. With the \textsc{Maximum} cost function, on the other hand, agent~$x$ prefers coalition $C_3$ over $C_1$, because the largest distance from $x$ to another agent is smaller in $C_3$ (distance $2$) than in $C_1$ (distance $3$). For the \textsc{Cutoff} cost function, agent $x$'s preference might depend on the threshold value $\lambda$ we set. When $\lambda=2$, all agents in coalition $C_3$ are within the threshold distance to agent $x$, while some agents in $C_1$ are outside. Thus, agent $x$ prefers $C_3$ over $C_1$ when $\lambda=2$. However, with a smaller threshold of $\lambda=1$, agent $x$ slightly prefers $C_1$ over $C_3$. See \Cref{sec:model} for the exact definition of the cost functions and \Cref{appx:intro} for details.
    \hfill $\vartriangleleft$
\end{example}
For our coalition formation problem with three natural cost functions and possibly a fixed number of coalitions~$k$, we establish the existence of stable coalition structures. Furthermore, we study structural properties of both stable and socially optimal states and analyze the quality of the stable states.

\subsection{Model}\label{sec:model}
\noindent\textbf{Agents and Coalition Structures.\,}
In our model there are $n$ agents denoted by $[n] \coloneqq \{1,2,\ldots,n\}$, each having a fixed real value $v_i \in \R$ (for all agents $i \in [n]$). We assume, wlog that the values are sorted in non-decreasing order, i.e., $L=v_1\leq v_2\leq \dots\leq v_n=H$.
Multiple agents may share the same value, which we call \emph{repetition}. 

Every subset of the agents is called a \emph{coalition}.
The coalition $[n]$ is called the \emph{grand coalition}, while a coalition $\{i\} \subseteq [n]$ consisting of a single agent $i \in [n]$ is called a \emph{singleton coalition}.
A set of coalitions $\mathcal{C}$ is called a \emph{coalition structure} if every agent is part of exactly one coalition in $\mathcal{C}$. For explicit coalition structures, we often represent an agent by its value. If a value $v \in \R$ or an entire coalition $\{a,b,c\}$ occurs $x \in \mathbb{N}$ times, we denote this by $$\overbrace{v\ldots v}^x, \text{ resp. }\overbrace{\{a,b,c\}}^x.$$ In this paper, we only focus on coalition structures consisting of at most $k \leq n$ many coalitions, and write $\mathcal{C} \coloneqq \{C_i\}_{i \in [k]}$, where $k \in \N$ is a model parameter.
Given a coalition structure $\mathcal{C}$, we denote the coalition of agent $i \in [n]$ as $\mathcal{C}(i)$. We say that a coalition structure~$\mathcal{C}$ 
(with possibly empty coalitions) is \emph{sorted}, if for all agents $i,j$ in some coalition $C \in \mathcal{C}$, it holds that all agents with a value between $v_i$ and $v_j$ are also part of the same coalition~$C$. In a sorted coalition structure, the coalitions can be sorted by the value of their agents. Then, w.l.o.g., we name the coalitions to be ascending with respect to the value of any agents in each coalition, with $C_1$ being the coalition with the agents with lowest value, etc. 

\vspace{6pt}
\noindent\textbf{Cost Functions.\,}
The \emph{distance} $\dist(i,j)$ between two agents $i,j \in [n]$ is defined as the Euclidean distance between the respective values of the agents, i.e., $|v_i-v_j|$. To calculate the cost of an agent $i$, there are many ways to aggregate the distances between it and a set of agents. 
In this paper, we explore three natural cost functions (average distance, maximum distance and cutoff), where for agent $i$, we aggregate its distances to \textit{all other agents in its current coalition} $C(i)$ while calculating its cost. Such cost functions are called \emph{without self-effect}. Because of this, singleton coalitions need to be treated separately: We differentiate between the \emph{happy in isolation (HIS)} and the \emph{unhappy in isolation (UIS)} case, where the cost of an agent $i \in [n]$ in a singleton coalition $\{i\}$ is zero or infinity, respectively. 

For coalitions with at least two agents, we define the three cost functions as follows.
The \textsc{Average} cost function
    $$\cost_\textrm{AVG}(i,C) \coloneqq \frac{1}{|C\setminus \{i\}|}\sum_{j \in C\setminus \{i\}}\dist(v_i,v_j)$$
takes the average distance to all other agents. The \textsc{Max} cost function 
    $$\cost_\textrm{MAX}(i,C) \coloneqq \max_{j \in C\setminus \{i\}}\dist(v_i,v_j)$$
takes the maximum distance to all other agents in the coalition. To define the \textsc{Cutoff} cost function, we need an additional model parameter $\lambda > 0$. Based on that, we partition the coalition $C$ into \emph{friends} and \emph{enemies} of agent $i$. The set of friends $N^+_i(C)$ of agent $i$ in coalition $C$ is the set of all other agents $j \in C \setminus \{i\}$ whose distance to agent $i$ is at most $\lambda$. All other agents in coalition $C$ are the \emph{enemies} $N^-_i(C) \coloneqq C \setminus (N^+_i(C) \cup \{i\})$ of agent $i$. The \textsc{Cutoff} cost function with respect to $\lambda > 0$ is then
    $$\cost_{\textrm{CUT}-\lambda}(i,C) \coloneqq \frac{|N_i^-(C)|}{|C \setminus \{i\}|},$$
i.e., the fraction of enemies of agent $i$ in coalition $C$ among all other agents in the coalition $C$. When it is clear from the context, we omit the subscript of the cost functions.

\vspace{6pt}
\noindent\textbf{Stability Concepts.\,}
Given a coalition structure $\mathcal{C}$, 
an agent $i$ 
has an \emph{improving move} if the cost of agent $i$ in its current coalition $C \in \mathcal{C}$ is strictly higher than its cost in some other coalition $C^\prime \in \mathcal{C} \setminus \{C\}$. A coalition structure is called \emph{jump stable} or a \emph{jump equilibrium} if no agent has an improving move. We refer to this  as \emph{jump stability}\footnote{In some related works this is called \emph{Nash stability}.}. 

Similarly, given a coalition structure $\mathcal{C}\coloneqq \{C_i\}_{i \in [k]}$, there is an \emph{improving swap} between agents $i$ and $j$ if both can strictly decrease their cost by swapping their coalitions, i.e., both leaving their current coalition and joining the other coalition. A coalition structure is called \emph{swap stable} if there is no improving swap.

\vspace{6pt}
\noindent\textbf{Terminology.\,}
We use the scheme \textsc{COST-ST-XIS} for naming combinations of cost functions and stability concepts, where 
\begin{itemize}
    \item \textsc{COST} indicates the cost function, which is either \textsc{Avg} (for Average), \textsc{Max} (for Maximum) or \textsc{Cutoff}$_\lambda$ (for the Cutoff cost function with parameter $\lambda$),
    \item \textsc{ST} indicates the stability concept, which is either \textsc{Jump} for jump stability or \textsc{Swap} for swap stability,
    \item \textsc{XIS} indicates whether we are in the happy in isolation (\textsc{HIS}) or unhappy in isolation (\textsc{UIS}) setting.
\end{itemize}

Note that, when considering swap stability, the sizes of the coalition cannot change. Therefore, singleton coalitions stay as they are. Moreover, every agent in a singleton coalition in the initial coalition structure cannot leave, since either it does not want to (HIS) or no other agent wants to swap into a singleton coalition (UIS). As singletons thus do not contribute to any game dynamic in the swap setting, we will only consider games without singletons and omit the \textsc{UIS} and \textsc{HIS} flag in the terminology for \textsc{Swap} games.

All game instances are defined by the number of agents $n$, their values $(v_i)_{i \in [n]}$ and the maximum number of allowed coalitions, $k$. Indeed, every instance of a \textsc{Jump} game is defined by the tuple $(n, (v_i)_{i \in [n]}, k)$. For \textsc{Swap} games, we additionally need to fix the sizes of the coalitions, i.e., every instance of a \textsc{Swap} game is defined by the tuple $(n, (v_i)_{i \in [n]}, k, (k_i)_{i \in [k]})$.

\vspace{6pt}
\noindent\textbf{Quality of Equilibria.\,}
The quality of a coalition structure $\mathcal{C}$ is measured by the \emph{social cost}, which we define as the sum of costs over all agent's cost in the game, i.e.,
    $\text{SC}(\mathcal{C}) \coloneqq \sum_{i \in [n]}\cost(i,\mathcal{C}(i))$.
A coalition structure with the smallest social cost is called an \emph{optimum} $\OPT(\mathcal{I})$ of a game instance $\mathcal{I}$. The quality of equilibria is measured by the Price of Anarchy (PoA) and the Price of Stability (PoS)~\cite{KoutsoupiasWorstcaseEquilibria2009,AnshelevichPriceStabilityNetwork2008}. Let $\mathcal{G}$ be a set of game instances and $\EQ$ the set of equilibria (given some stability concept). Then the PoA (PoS) is defined to be the worst ratio between the cost of the worst (best) equilibrium and the optimum cost over all instances of $\mathcal{G}$, i.e.,
\begin{align*}
    \PoA(\mathcal{G}) \coloneqq \sup_{\mathcal{I} \in \mathcal{G}}\sup_{\mathcal{C} \in \EQ(\mathcal{I})}\frac{\SC(\mathcal{C})}{\OPT(\mathcal{I})},\\
    \PoS(\mathcal{G}) \coloneqq \sup_{\mathcal{I} \in \mathcal{G}}\inf_{\mathcal{C} \in \EQ(\mathcal{I})}\frac{\SC(\mathcal{C})}{\OPT(\mathcal{I})}.
\end{align*}

Regarding game dynamics, 
a game has the \textit{finite improvement property} (FIP), if any sequence of improving moves is finite, i.e., eventually must end in an equilibrium. 
We will use \textit{ordinal potential functions}~\cite{MondererPotentialGames1996} to prove this property in a setting, or show the existence of an \textit{improving response cycle} (IRC) to disprove it.

To discuss the PoS for \textsc{Cutoff$_\lambda$-Jump} games, we use the technical concept of a $\lambda$-block. Given some $\lambda > 0$, a \emph{$\lambda$-block} is an interval $I \subseteq \R$ of size $|I|=\lambda$. A (multi-)set of values $V$ is said to be \emph{covered} by a set $\mathcal{B}$ of $\lambda$-blocks, if for every value $v \in V$ there is at least one $\lambda$-block $B \in \mathcal{B}$ that includes the values $v$, i.e., $v \in B$. We call a \textsc{Cutoff$_\lambda$-Jump} game instance $(n,(v_i)_{i\in [n]},k)$ \emph{nice} if the set of values $\{v_i\}_{i \in [n]}$ can be covered by a set of at most $k$ $\lambda$-blocks.

\begin{table*}[!ht]
\centering
\begin{NiceTabular}{|c||c|c|c||c|c|c||c|c|c|}[hvlines]
    \Block{2-1}{ } & 
    \Block{1-3}{\textsc{Average}} & & & 
    \Block{1-3}{\textsc{Cutoff$_\lambda$}} & & & 
    \Block{1-3}{\textsc{Max}} \\
    & 
    \textsc{Swap} & \textsc{J-HIS} & \textsc{J-UIS} & 
    \textsc{Swap} & \textsc{J-HIS} & \textsc{J-UIS} & 
    \textsc{Swap} & \textsc{J-HIS} & \textsc{J-UIS} \\ \Hline \Hline
    EQ Existence & 
    \doublecheckmark C\ref{cor:swap-potential-games} & \checkmark T2.4~\cite{MilchtaichStabilitySegregationGroup2002} & \checkmark GC & 
    \doublecheckmark C\ref{cor:swap-potential-games} & \checkmark T\ref{theorem:his-sorted-pne} & \checkmark GC & 
    \doublecheckmark T1~\cite{BiloSchellingGamesContinuous2023} & \Block{1-2}{\doublecheckmark L\ref{lemma:max-jump-fip}}\\
    \Block{2-1}{EQ sorted} & 
    
    $\forall\exists$ L\ref{theorem:swap-sorted-eq-exists} & 
    $\forall\exists$ T2.4~\cite{MilchtaichStabilitySegregationGroup2002}& 
    $\forall\exists$ GC & 
    
    \;$\forall\exists$ L\ref{theorem:swap-sorted-eq-exists}\; & 
    \;\;$\forall\exists$ T\ref{theorem:his-sorted-pne}\;\; & 
    $\;\;\;\forall\exists$ GC\;\;\; & 
    
    $\forall\exists$ T\ref{theorem:swap-sorted-eq-exists} & 
    $\forall\exists$ T\ref{theorem:his-sorted-pne} & 
    $\forall\exists$ GC \\
    & 
    
    $\existsnotall$ L\ref{lemma:swap-unsorted-opt} & \Block{1-2}{$\existsnotall$ E2.3~\cite{MilchtaichStabilitySegregationGroup2002}}&&
    
    $\existsnotall$ L\ref{lemma:swap-unsorted-opt} & \Block{1-2}{$\existsnotall$ L\ref{lemma:cutoff-jump-unsorted}}&&
    
    $\existsnotall$ L\ref{lemma:swap-unsorted-opt} & \Block{1-2}{$\existsnotall$ L\ref{lemma:PoA-Max-Avg-Jump-Swap}}&\\
    OPT sorted & 
    
    $\exists{\not\exists}$ L\ref{lemma:swap-unsorted-opt} & 
    \Block{1-2}{($\forall\exists$) \textit{Conjecture}} & & 
    
    \Block{1-3}{$\exists{\not\exists}$ L\ref{lemma:swap-unsorted-opt}} & & &
    
    $\exists{\not\exists}$ L\ref{lemma:swap-unsorted-opt} & 
    \Block{1-2}{$\exists{\not\exists}$ L\ref{lemma:unsorted-unstable-max-jump-opt}} & 
    \\
    PoA & 
    
    \Block{1-3}{$\infty$ L\ref{lemma:PoA-Max-Avg-Jump-Swap}} & & & 
    
    \Block{1-3}{$\infty$ L\ref{lemma:PoA-cutoff}} & & &  
    
    \Block{1-3}{$\infty$ L\ref{lemma:PoA-Max-Avg-Jump-Swap}} & & \\
    PoS & 
    
    =$1$ C\ref{cor:cutoff-avg-swap-pos} & 
    \Block{1-2}{${>}1$ L\ref{lemma:avg-max-jump-pos}} & & 
    
    =$1$ C\ref{cor:cutoff-avg-swap-pos} & 
    \Block{1-2}{${=}1$ (nice) / ${>}1$ (else) L\ref{lemma:pos-cutoff-jump}} &
    &
    
    =$1$ L\ref{lemma:max-swap-pos} &
    \Block{1-2}{${>}1$ L\ref{lemma:avg-max-jump-pos}} &
\end{NiceTabular}
\caption{Result Overview. The symbol \enquote{$\doublecheckmark$} means that this setting has the FIP, while \enquote{$\checkmark$} only indicates simple equilibrium existence for all instances (e.g. GC stands for grand coalition). As for the logic quantifier, the first is over the set of instances of this setting while the second one is over the set of equilibria of the respective instance, i.e., \enquote{$\existsnotall$} can be read as \enquote{There is an instance such that not all equilibria are sorted}; \enquote{$\forall\exists$} means \enquote{For all instances there is at least one sorted equilibrium / optimum}; \enquote{$\exists\not\exists$} means \enquote{There are instances where no optimum is sorted}. The \enquote{$\infty$} in the PoA row indicates an unbounded~PoA.\label{table:results}}
\end{table*}

\subsection{Related Work}
Given the breadth of variants of hedonic games, we focus our discussion on models with \emph{cardinal}-based comparisons, where each agent has a valuation function scoring every coalition and prefers the higher (lower) valued coalition in the utility (cost) variant. 
However, some models based on ordinal comparisons employ ideas similar to the cost functions we study but their results cannot be applied to our model. Examples include: hedonic games with $\mathcal{W}$-preferences~\cite{AzizIndividualbasedStabilityHedonic2012,CechlarovaStablePartitionsWpreferences2004,CechlarovaStabilityCoalitionFormation2001}, where coalitions are compared based on the worst-ranked agent (similar to our \textsc{Max} cost function), and hedonic games with friends-and-enemies preferences~\cite{DimitrovSimplePrioritiesCore2006,LangRepresentingSolvingHedonic2015,SchlueterFriendEnemyorientedHedonic2024}, where each agent partitions others into friends and enemies (similar to our \textsc{Cutoff} cost function) and compares coalitions based on the number of friends and enemies they include.

The most prominent cardinal-based hedonic games with succinct representations are \emph{additively separable hedonic games (ASHGs)} introduced by \citet{BogomolnaiaStabilityHedonicCoalition2002} and \emph{fractional hedonic games (FHGs)} introduced by \citet{AzizFractionalHedonicGames2019}. In both games, every agent $i$ has a valuation $v_i(j) \in \R$ for every agent $j$ in the game. In ASHGs, the utility of an agent is the sum of these valuations over all agents in its coalition, while in FHGs, the agents take the average of the values in its coalition.  
Recently, \citet{MonacoStableOutcomesModified2020} studied a variant of FHGs called \emph{modified FHGs} where every agent takes the average only over all \emph{other} agents' valuations, excluding its own (without self-effect). Indeed, the \textsc{Avg-Jump-UIS} and \textsc{Cutoff$_\lambda$-Jump-UIS} games can be seen as cost-based variants of these modified FHGs provided that the number of coalitions is not restricted. 
For the \textsc{Cutoff$_\lambda$} games, the valuations are additionally in $\{0,1\}$, which is why they are a subclass of so called \emph{unweighted} modified FHGs. Also, \citet{MonacoStableOutcomesModified2020} showed that in modified FHGs with non-negative and symmetric valuations, there is always a jump stable outcome (the grand coalition) and the PoA and PoS is linear in $n$. However, these results do not hold for our cost-based variants; especially when restricting the number of coalitions.

Several hedonic game variants use some notion of distance in their models. \citet{ReyFENHedonicGamesDistanceBased2022} use the distance between an agent's ordinal preference list and its coalition to determine their valuation of that coalition. Further, some variants of FHGs position the agents on an unweighted host graph $G$ and use the distance in $G$ to get the valuations between agents. Examples include \emph{distance hedonic games} (DHGs) introduced by \citet{FlamminiDistanceHedonicGames2021} and \emph{social distance games} introduced by \citet{BranzeiSocialDistanceGames2011}. In social distance games (SDGs), the valuation of an agent $i$ for an agent $j \neq i$ in a coalition $C$ is $\frac{1}{d_{G[C]}(i,j)}$, where $d_{G[C]}(i,j)$ is the distance from node $i$ to node $j$ in the subgraph of $G$ induced by the node-set~$C$. Distance hedonic games (DHGs) generalize SDGs by adding a scoring vector $\alpha$ mapping each possible distance $d_{G[C]}(i,j)$ to a real number. Some instances of our \textsc{Avg-Jump} game can be modeled as DHGs. Here, agents have consecutive natural numbers as fixed values and the matching DHG instance is played on a path graph with a decreasing scoring vector $\alpha$. To the best of our knowledge, there no results are known for this case.

Only a few models consider distances between points embedded in the Euclidean space. In \emph{hedonic clustering games} introduced by \citet{FeldmanHedonicClusteringGames2015}, coalitions are interpreted as clusters, and each agent’s cost is defined by its distance to the cluster center or median. 
Also, they examine the \enquote{Fixed Clustering} setting, in which the number of coalitions (or clusters) is fixed. Both variants are related to our model, but incomparable due to the different cost functions.

Regarding limiting the number of coalitions, \citet{SlessFormingCoalitionsFacilitating2018} study ASHGs where exactly $k$ coalitions must be formed. \citet{LiFractionalHedonicGames2021} extends this by considering FHGs with an upper bound on the number of coalitions. 
\citet{MilchtaichStabilitySegregationGroup2002} analyze a model incorporating all three key aspects of our models: limited number of coalitions, geometric embedding, and ignoring the self-effect within coalitions. In particular, they study exactly the \textsc{Avg-Jump-HIS} games and show that sorted 
PNE always exist, although unsorted PNE can also occur, and an IRC exists even if there are only two coalitions. In order to show the existence of sorted PNE, they propose an algorithm that we will extend and adapt to our other cost functions. Interestingly, they also consider the Distance-to-Average cost function 
in higher dimensions, where the existence of PNE is no longer guaranteed.

A model similar to our \textsc{Avg-Jump-HIS} game is also analyzed by \citet{Ahmadi2022}. For the general setting in multiple dimensions, they show NP-hardness and for the one-dimensional average setting with Euclidean distances, they rediscover the $\bigO(kn)$ algorithm by \citet{MilchtaichStabilitySegregationGroup2002}.
Their model was later generalized by \citet{Aamand2023}, which allows for arbitrary distance aggregation functions for the \textsc{Jump-HIS} setting.
With this, they investigate a minimum distance and maximum distance function, similar to our \textsc{Max-Jump-HIS} setting, where they show approximation results on the optimum in a multidimensional setting.

Another way to limit the number of coalitions is by fixing the size of each coalition. As jumps would change coalition sizes, swap stability is more applicable in this setting. However, we note that some studies on hedonic games consider upper bounds on coalition sizes, which allows jumps~\cite{BiloHedonicGamesFixedsize2022,CsehParetoOptimalCoalitions2019}. Swap stability has been studied in various models including matching markets~\cite{AlcaldeExchangeproofnessDivorceproofnessStability1994,AzizCoalitionalExchangeStable2017,DamammePowerSwapDeals2015}, Schelling Games~\cite{ChauhanSchellingSegregationStrategic2018,EchzellConvergenceHardnessStrategic2019,AgarwalSchellingGamesGraphs2021,KanellopoulosModifiedSchellingGames2021,BiloTopologicalInfluenceLocality2022}, and also recently in hedonic games with $\mathcal{W}$-preferences~\cite{CsehParetoOptimalCoalitions2019} and ASHGs~\cite{BiloHedonicGamesFixedsize2022}.
Schelling Games~\cite{ChauhanSchellingSegregationStrategic2018,AgarwalSchellingGamesGraphs2021}, in particular, can be seen as hedonic games with overlapping coalitions. In these games, agents of different types occupy nodes on a host graph and aim to optimize their neighborhood, which corresponds to their coalition in this setting. Formally, hedonic games are Schelling games on a host graph consisting of disconnected cliques, while preserving the same cost function and stability concept. We highlight the model by \citet{BiloSchellingGamesContinuous2023}, as they study the same set of cost functions and stability concepts as those used in our hedonic games. Using the reduction mentioned above, we obtain that our \textsc{Max-Swap} games are potential games. Furthermore, they show that the PoA is unbounded in almost all cases except for \textsc{Max-Swap} and \textsc{Avg-Swap}, where the PoA is linear in the number of agents. They also show that the PoS is $1$ for \textsc{Avg-Swap} and \textsc{Cutoff$_\lambda$-Swap} games on regular graphs. Note that the latter result does not directly apply to our models, as the reduction to Schelling games does not create regular and connected graphs as used by~\citet{BiloSchellingGamesContinuous2023}.

\subsection{Our Contribution}
We propose and analyze a novel succinct variant of hedonic games where each agent has a fixed real-valued type and costs are induced by type-distances. To capture natural settings where agents cannot opt out to form a singleton coalition, the maximum number of coalitions might be fixed.   
For an overview of our results with the respective references, see \Cref{table:results}.

For all combinations of cost functions (average distance, distance cutoff, maximum distance; in columns) and stability concepts (swap and jump; in sub-columns), we show that: every game instance admits at least one equilibrium (table row 1), either because they fulfill the FIP, or as we show a way to construct one for any instance.
Further, every game instance admits a sorted equilibrium (row 2) and there exists at least one game instance in each setting admitting an unsorted equilibrium (row 3). 
Furthermore, we identify game instances in which all socially optimal structures are unsorted in every setting (row 4), except for \textsc{Avg-Jump}, where we conjecture that all optimal structures are sorted. 
To substantiate this, we provide partial results in this direction~(\Cref{appx:avg-jump-opt}).
Regarding the quality of equilibria, we provide examples showing that the price of anarchy (PoA) is unbounded for most settings (row 5), and we characterize the price of stability (PoS) for all \textsc{Swap} games.
For \textsc{Jump} games, the PoS is strictly greater than $1$, except for a specific subclass of \textsc{Cutoff$_\lambda$-Jump} games (row~6).

In \Cref{sec:swap,sec:jump}, we address swap stability and jump stability, separately.
Each section first establishes the existence of equilibria and then investigates their structural properties, along with those of the optimal structures.
Finally, in \Cref{sec:eq-quality}, we analyze the quality of both swap and jump equilibria in terms of the PoA and PoS.

\section{Swap Stability}\label{sec:swap}
In this section, we study the existence and properties of swap equilibria as well as the properties of the respective optima. 

\subsection{Equilibrium Existence}
In this subsection, we show that all \textsc{Swap} games have at least one equilibrium. \citet{BiloSchellingGamesContinuous2023} give potential functions for all three cost functions for Schelling games with continuous types on regular graphs (Theorem 2 and Proposition 3) and for \textsc{Max-Swap} even on general graphs (Theorem 1). As our hedonic games can be transformed to their Schelling games on a set of (different-sized) cliques, we already have equilibrium existence for \textsc{Max-Swap}. For \textsc{Avg-Swap} and \textsc{Cutoff$_\lambda$-Swap}, they showed that the social cost is a potential function. \citet{BiloHedonicGamesFixedsize2022} showed that the social welfare (the utility equivalent of social cost) is also a potential function for ASHG and FHG restricted for swap stability. In the following \Cref{lemma:swap-mfhg}, we show that this also applies to a subclass of modified FHGs~\cite{MonacoStableOutcomesModified2020} that includes utility variants for our \textsc{Avg} and \textsc{Cutoff$_\lambda$} cost functions. The proof can be found in~\Cref{appx:swapStability}.

\begin{restatable}{lemma}{lemmaSwapMFHG}
\label{lemma:swap-mfhg}
Every modified FHG with symmetric and non-negative valuation has a swap equilibrium, and they has the FIP.
\end{restatable}

Using \Cref{lemma:swap-mfhg} for \textsc{Avg-} and \textsc{Cutoff$_\lambda$-Swap} games, as well as the potential function for \textsc{Max-Swap} Schelling games by~\citet{BiloSchellingGamesContinuous2023}, we prove \Cref{cor:swap-potential-games} in \Cref{appx:swapStability}.

\begin{restatable}{corollary}{corSwapPotentialGames}
\label{cor:swap-potential-games}
Every \textsc{Max}-, \textsc{Avg}- and \textsc{Cutoff$_\lambda$-Swap} game has an equilibrium, and all three fulfill the FIP.
\end{restatable}

\subsection{Properties of Equilibria and Optima}
In this section, we study whether or when swap equilibria and optima of \textsc{Swap} games are sorted. First, we show the intuitive statement, that sorted coalition structures are swap stable for all three cost functions (see \Cref{appx:swapStability} for the proof).

\begin{restatable}{theorem}{theoremSwapSortedEqExists}
\label{theorem:swap-sorted-eq-exists}
Every sorted coalition is swap stable for \textsc{Max}, \textsc{Avg} and \textsc{Cutoff$_\lambda$} games.
\end{restatable}

Although \Cref{theorem:swap-sorted-eq-exists} shows that all sorted coalition structures are swap stable, this does not necessarily hold in the other direction.
In \Cref{lemma:swap-unsorted-opt}, we give examples indicating that neither swap equilibria nor (swap) optima are always sorted. 

\begin{restatable}{lemma}{lemmaSwapUnsortedOpt}
\label{lemma:swap-unsorted-opt}
    For \textsc{Max}-, \textsc{Avg}- and \textsc{Cutoff$_\lambda$-Swap} games, there is an instance with only unsorted optima and an unsorted equilibrium.
\end{restatable}
\begin{proof}
    Consider the following game instance with $n \coloneqq9$ agents, $k\coloneqq2$ coalitions of size $k_1\coloneqq4$ and $k_2\coloneqq5$ and values $(v_i)_{i \in [9]}\coloneqq(1,1,2,2,2,2,2,3,3)$. For all cost functions, we show in \Cref{appx:swapStability} that the unsorted coalition structure
    \begin{align*}
        \mathcal{C}^* \coloneqq  \bigl \{\{1,1,3,3\},\{2,2,2,2,2\} \bigr \}
    \end{align*}
    is the social optimum and a swap equilibrium. For the \textsc{Cutoff$_\lambda$} cost function, we choose $\lambda < 1$. However, one can scale the values of this example to get the same results for an arbitrary $\lambda$. 
\end{proof}

\section{Jump Stability}\label{sec:jump}
\Cref{sec:swap} shows that swap equilibria are easy to find due to several potential functions. In this section, we show that equilibria are harder to find for \textsc{Jump} games, since the FIP only holds for \textsc{Max-Jump} games, and sorted coalition structures as well as social optima are not always stable.

\subsection{Equilibrium Existence}
First, note that in the unhappy in isolation (UIS) setting, the grand coalition is a stable coalition structure, as no agent would jump out into a singleton coalition. Thus, equilibrium existence is given for all cost functions in this case.

\begin{observation}\label{obs:Jump-UIS-grand-coalition-stable}
    For all three \textsc{Jump-UIS} game variants, the grand coalition is always stable.
\end{observation}

Moreover, \textsc{Max-Jump-UIS} games and also \textsc{Max-Jump-HIS} games are potential games, which we show in the following \Cref{lemma:max-jump-fip}. \citet{BiloSchellingGamesContinuous2023} show that the potential function for \textsc{Max-Swap} in Schelling games with continuous types fails on regular graphs where the number of empty nodes is higher than the maximal degree in the host graph. However, as \Cref{lemma:max-jump-fip} implies, this potential function works on any set of independent cliques in the \textsc{HIS} setting, and with minor adjustments also for the UIS setting.

\begin{restatable}{lemma}{lemmaMaxJumpFIP}
\label{lemma:max-jump-fip}
    Every \textsc{Max-Jump-UIS}, and \textsc{-HIS} game has an equilibrium, and they fulfill the FIP.
\end{restatable}
\begin{proof}[Proof sketch]
    We show that both kind of games are potential games. Both potential functions are based on the idea that the non-increasingly sorted cost vector of all agents decreases lexicographically with every improving jump. Formally, let $\Phi_{\cost}(\mathcal{C})$ be the vector of the costs of all agents in the HIS setting that is sorted non-increasingly. Then, we show that $\Phi_{\cost}(\mathcal{C})$ is a potential function for \textsc{Max-Jump-HIS} games and that $\Phi(\mathcal{C}) \coloneqq (|\mathcal{C}|_{\neq \emptyset},\Phi_{\cost}(\mathcal{C}))$ is a potential function for \textsc{Max-Jump-HIS} games, where $|\mathcal{C}|_{\neq \emptyset}$ is the number of non-empty coalitions in the given coalition structure~$\mathcal{C}$. Details can be found in \Cref{appx:JumpStability}. 
\end{proof}

Next, we provide an IRC for \textsc{Avg-} and \textsc{Cutoff$_\lambda$-Jump} games in \Cref{theorem:avg-cutoff-jump-not-potential} showing that both are not potential games in comparison to the \textsc{Max-Jump} games. Note that, \citet{MilchtaichStabilitySegregationGroup2002} provided an IRC for \textsc{Avg-Jump} games. However, we want to highlight that the IRC given in \Cref{theorem:avg-cutoff-jump-not-potential} works for both \textsc{Avg-} and \textsc{Cutoff$_\lambda$-Jump} games.

~

\begin{theorem}\label{theorem:avg-cutoff-jump-not-potential}
    \textsc{Avg-} and \textsc{Cutoff$_\lambda$-Jump} games admit an IRC.
\end{theorem}
\begin{proof}
    Consider the improving response cycle given in \Cref{table:jump-ircs}, which works for both cost functions. Please note that this IRC also holds for both the \textsc{HIS} and the \textsc{UIS} setting, as it does not include singleton coalitions.
\end{proof}

\begin{table}[!h]
    \begin{center}
        \begin{tabular}{ c | c | l}
        Coalition Structure & \textsc{Cutoff$_4$} & \textsc{Average} \\ \hline
        $\{14,11,5,6,7,9\},\,\{1,\cev{\mathbf{5}},5,8,10,14\}$ & $\frac{2}{5} \leadsto \frac{2}{6}$ & $\frac{21}{5} \leadsto \frac{22}{6}$ \\
        $\{14,11,5,5,6,7,9\},\,\{\cev{\mathbf{1}},5,8,10,14\}$ & $\frac{3}{4} \leadsto \frac{5}{7}$ & $\frac{33}{4} \leadsto \frac{50}{7}$ \\
        $\{14,11,1,5,5,6,7,\vec{\mathbf{9}}\},\,\{5,8,10,14\}$ & $\frac{2}{7} \leadsto \frac{1}{4}$ & $\frac{28}{7} \leadsto \frac{13}{5}$ \\
        $\{14,11,1,5,5,6,\vec{\mathbf{7}}\},\,\{5,8,9,10,14\}$ & $\frac{2}{6} \leadsto \frac{1}{5}$ & $\frac{22}{6} \leadsto \frac{15}{5}$ \\
        $\{14,11,1,5,\vec{\mathbf{5}},6\},\,\{5,7,8,9,10,14\}$ & $\frac{2}{5} \leadsto \frac{2}{6}$ & $\frac{20}{5} \leadsto \frac{23}{6}$ \\
        $\{14,11,\vec{\mathbf{1}},5,6\},\,\{5,5,7,8,9,10,14\}$ & $\frac{3}{4} \leadsto \frac{5}{7}$ & $\frac{33}{4} \leadsto \frac{40}{7}$ \\
        $\{14,11,5,6\},\,\{1,5,5,7,8,\cev{\mathbf{9}},10,14\}$ & $\frac{2}{7} \leadsto \frac{1}{4}$ & $\frac{25}{7} \leadsto \frac{14}{4}$ \\
        $\{14,11,5,6,9\},\,\{1,5,5,\cev{\mathbf{7}},8,10,14\}$ & $\frac{2}{6} \leadsto \frac{1}{5}$ & $\frac{21}{6} \leadsto \frac{16}{5}$ \\
        $\{14,11,5,6,7,9\},\,\{1,\cev{\mathbf{5}},5,8,10,14\}$ & \dots & \\
        \end{tabular}
    \end{center}
    \caption{IRC for the given game instance under the \textsc{Cutoff$_4$-Jump} setting, and \textsc{Avg-Jump} setting. The first column shows the coalition structures of the IRC with the jumping agent highlighted in bold. The other two columns indicate the cost change of the jumping agent.\label{table:jump-ircs}}
\end{table}

Although we know by \Cref{theorem:avg-cutoff-jump-not-potential} that \textsc{Avg-Jump} and \textsc{Cutoff$_\lambda$-Jump} are not potential games, there is at least one jump equilibrium for each such game instance. Indeed, \citet{MilchtaichStabilitySegregationGroup2002} gave an algorithm that calculates a jump stable coalition structure for all \textsc{Avg-Jump-HIS} game instances that is also sorted \cite[Theorem 2.4]{MilchtaichStabilitySegregationGroup2002}. In the following, we show that this algorithm also works for a large class of games, including \textsc{Avg-}, \textsc{Cutoff$_\lambda$-} and \textsc{Max-Jump-HIS} games.

In order to explain and analyze this algorithm, we introduce some additional notation and terminology regarding a sorted coalition structure $\mathcal{C}\coloneqq\{C_i\}_{i \in [k]}$. For every coalition $C_i \in \mathcal{C}$, we call the agent with the smallest and highest value in $C_i$ the \emph{left-most} agent $L(C_i)$ and the \emph{right-most} agent $R(C_i)$ of $C_i$, respectively. Having the same idea in mind, we also refer to coalition $C_{i-1}$ and $C_{i+1}$ as the coalitions \emph{left} and \emph{right} of coalition $C_i$, respectively\footnote{Note that there is no coalition left of coalition $C_1$ nor right of coalition $C_k$.}. If some left-most agent $x \coloneqq L(C_i)$ wants to move to the next coalition to its left, i.e., $C_{i-1}$, to reduce its cost, we call this move a \emph{left-improving move} of agent $x$. Similarly, a \emph{right-improving move} is a move of some right-most agent $x \coloneqq R(C_i)$ to coalition $C_{i+1}$ such that agent~$x$ has higher cost in $C_i$ than in $C_{i+1}$. 

The idea of Milchtaich and Winter's algorithm is to start with a sorted \enquote{right-heavy} coalition structure, having the $k-1$ smallest values of a game instance in singleton coalitions and the rest in the remaining coalition. Then the algorithm allows the agents to perform left-improving moves and stops if no agent has a left-improving move. \citet{MilchtaichStabilitySegregationGroup2002} show that, with the \textsc{Average} cost function, that the algorithm terminates after at most $nk$ many left-improving moves in jump stable coalition structure. In \Cref{def:monotone-cost-f}, we bundle sufficient properties of the \textsc{Average} cost function into a broad class of cost function inducing the same behavior to the algorithm.

\begin{definition}\label{def:monotone-cost-f}
    A cost function $\cost_m$ is called \emph{monotone}, if the following conditions hold for every game instance $(n,(v_i)_{i\in [n]},k)$:
    \begin{enumerate}
        \item[(i)] For a coalition $C$ and agents $x,y \notin C$:\\
            if $v_x \leq v_y \leq v_{L(C)}$ or $v_{R(C)} \leq v_y \leq v_x$ \\
            then $\cost_m(y,C) \leq \cost_m(x,C)$.
        \item[(ii)] For all coalitions $C,D$ and agents $x \notin C \cup D$: \\
        if $v_x{\leq}v_{L(C)}{\leq}v_{R(C)}{\leq}v_{L(D)}$ or $v_{R(D)}{\leq}v_{L(C)}{\leq}v_{R(C)}{\leq}v_x$\\
        then $\cost_m(x,C) \leq \cost_m(x,C\cup D)\leq \cost_m(x,D)$.
        \item[(iii)] For all coalitions $C,D$ over $\mathbb{R}$ and $c \in C$:\\
            if $\cost_m(c,C) > \cost_m(c,D)$, then
            \begin{itemize}
                \item $\cost_m \bigl(R(C),C \bigr) > \cost_m \bigl(R(C),D \bigr)$ $\left(\text{for }v_{R(C)} \leq v_{L(D)}\right)$,
                \item $\cost_m \bigl(L(C),C \bigr) > \cost_m \bigl(L(C),D \bigr)$ $\left(\text{for } v_{R(D)} \leq v_{L(C)}\right).$
            \end{itemize}
    \end{enumerate}
\end{definition}

In order to prove that the algorithm by \citet{MilchtaichStabilitySegregationGroup2002} calculates a jump stable coalition structure for all monotone cost functions (see \Cref{theorem:his-sorted-pne}), we first show that we only need to rule out left- and right-improving moves to prove stability for a given sorted coalition structure. This property is captured in the following \Cref{lemma:left-right-improving-jump} and can be proven as an implication of property (iii) of \Cref{def:monotone-cost-f} (see \Cref{appx:JumpStability} for details).

\begin{restatable}{lemma}{lemmaImprovingMoveStable}
    \label{lemma:left-right-improving-jump}
    Let $\mathcal{I} \coloneqq (n,(v_i)_{i\in [n]},k)$ be a game instance with monotone cost function $\cost_m$. Further, let $\mathcal{C}\coloneqq \{C_i\}_{i \in [k]}$ be a sorted coalition structure of $\mathcal{I}$. Then, either $\mathcal{C}$ is a jump equilibrium, or there is a left-improving move or a right-improving move.
\end{restatable}

Now, we are ready to prove that the properties of monotone cost functions are sufficient for the algorithm of \citet{MilchtaichStabilitySegregationGroup2002} to work properly.

\begin{restatable}{theorem}{theoremHisSortedPne}\label{theorem:his-sorted-pne}
    Every game instance $\mathcal{I} \coloneqq (n,(v_1,\dots,v_n),k)$ with a monotone cost function $\cost_m$ in the \textsc{HIS}-setting has a sorted jump equilibrium, which can be found within $kn$ many improving moves.
\end{restatable}
\begin{proof}[Proof sketch]
    The algorithm starts with the \enquote{right-heavy} coalition structure 
        $$\mathcal{C}_{\text{start}} \coloneqq \big\{\{v_1\},\dots,\{v_{k-1}\},\{v_k,\dots,v_n\}\big\}.$$
    Clearly, $\mathcal{C}_{\text{start}}$ is sorted and  stays sorted under left-improving moves. It is therefore valid to also number coalitions during the algorithm from $C_1$ to $C_k$ according to the ordering of their values. 
    
    To prove that this algorithm works under the monotone cost function $\cost_m$, we show by induction that there is no coalition structure during the algorithm admitting a right-improving move. By \Cref{lemma:left-right-improving-jump}, this proves that the final coalition structure of this algorithm, for which by definition no left-improving move is possible, is indeed a jump equilibrium.

    In the base case, $\mathcal{C}_{\text{start}}$ is sorted due to the assumed \textsc{HIS} setting. Now consider a sorted coalition structure $\mathcal{C}\coloneqq \{C_i\}_{i \in [k]}$ which does not admit any right-improving move, but at least one left-improving move. Let $L(C_i)$ be the agent with a left-improving move in $\mathcal{C}$, and $\mathcal{C}^\prime \coloneqq \{C^\prime_i\}_{i \in [k]}$ the coalition structure $\mathcal{C}$ after the left-improving move of agent $L(C_i)$. See \Cref{fig:cutoff-pne-proof} for an illustration. 

    \begin{figure}[h]
        \centering
        \includegraphics[width=\linewidth]{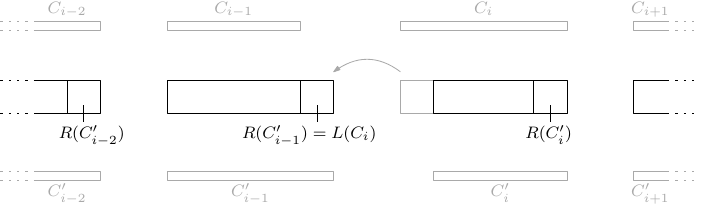}
        \caption{Sketch of a left-improving move from a coalition structure $\{C_i\}_{i \in [k]}$ to $\{C^\prime_i\}_{i \in [k]}$.}
        \label{fig:cutoff-pne-proof}
        \Description{
        Agents are grouped as their coalitions from left to right. The coalitions $C_{i-2}$,$C_{i-1}$,$C_{i}$ are depicted, as well as their right most agents $R(C_j)$. The left most agent $L(C_i)$ is shown as makes a left-improving move to $C_{i-1}$, thereby becoming $R(C'{i-1})$. 
        }
    \end{figure}

    The only coalitions that changed are $C_i$ and $C_{i-1}$. Therefore, we only need to show that none of the agents $R(C^\prime_{i-2})$, $R(C^\prime_{i-1})$ and $R(C^\prime_i)$ have a right-improving move. For complete proof, see~\Cref{appx:JumpStability}.
\end{proof}

Next, we show that all three cost functions are indeed monotone. This implies that all \textsc{Jump-HIS} game instances have at least one sorted equilibrium. For the proof of~\Cref{lemma:monotone-cost-f}, refer to \Cref{appx:JumpStability}.

\begin{restatable}{lemma}{lemmaMonotoneCostF}\label{lemma:monotone-cost-f}
    The \textsc{Max-}, \textsc{Average-} and \textsc{Cutoff$_\lambda$} cost function are all monotone.
\end{restatable}

\subsection{Properties of Optima and Equilibria}
First, let us observe that if we fix $n$, $k$, and the size of each coalition, the optimum structure for both \textsc{Jump} and \textsc{Swap} games is the same, as social costs are computed in the same way. 
But as the game dynamics of \textsc{Jump} games allow the coalition number and sizes to change dynamically, we consider optima only depending on $n$ and maximum $k$.
Especially for \textsc{Jump}-\textsc{HIS} games, the stable coalition structures will always have $k$ occupied coalitions, as singleton resp. single-value coalitions will always incur the minimum cost of 0 for each contained agent, and thus instances with a higher maximum $k$ may have a lower optimum social cost.
In contrast to this, in the \textsc{UIS} setting, agents in singleton coalitions will always leave them and never enter empty coalitions, and thus coalitions may stay empty, and the grand coalition will always be an equilibrium.

The following result has already been mentioned but not proven by \citet[Example 2.3]{MilchtaichStabilitySegregationGroup2002}. The proof and further observations can be found in \Cref{appx:JumpStability}.

\begin{restatable}{lemma}{lemmaAvgUnsortedJumpEq}
\label{lemma:avg-unsorted-jump-eq}
    For \textsc{Avg}-\textsc{Jump} games, there is an instance with at least one unsorted equilibrium.
\end{restatable}

\begin{restatable}{lemma}{lemmaCutoffjumpunsorted}\label{lemma:cutoff-jump-unsorted}
    For \textsc{Cutoff$_\lambda$}-\textsc{Jump} games, there is an instance with only unsorted optima and at least one unsorted equilibrium.
\end{restatable}
\begin{proof}
    We use a similar construction as in \Cref{lemma:swap-unsorted-opt}, but this time with only three agents with value $2$: Consider the game instance with $n \coloneqq7$ agents, $k\coloneqq2$, some $\lambda < 1$ and values $(v_i)_{i \in [7]}\coloneqq(1,1,2,2,2,3,3)$. The unsorted coalition structure
    \begin{align*}
        \mathcal{C}^* \coloneqq  \bigl \{\{1,1,3,3\},\{2,2,2\} \bigr \}
    \end{align*}
    is jump stable, as all but one coalition have only pairwise friends in it. In \Cref{appx:JumpStability}, we show that it is also the social optimum.
\end{proof}

For \textsc{Max}-\textsc{Jump} games, there is a simple construction to show the existence of unsorted equilibria.

\begin{restatable}{lemma}{lemmaMaxJumpUnsortedEq}\label{lemma:max-jump-unsorted-eq}
    For \textsc{Max}-\textsc{Jump-\{UIS,HIS\}} games and any $k\geq 2$ and $n\geq 4k$, there is an instance with at least one unsorted equilibrium.
\end{restatable}
\begin{proof}[Proof sketch]
    Here, we only give the worst-case example for $k=2$ and $n=4$ and refer to \Cref{appx:JumpStability} for the general construction. For this, consider the following instance 
        $$\big\{\{L,L,H,H\},\{L,L,H,H\}\big\},$$
    where no agent has an improving jump, as its valuation is $H-L$ in both coalitions.
\end{proof}

By adjusting the instance used in \Cref{lemma:swap-unsorted-opt}, we get that \textsc{Max} games also admit unsorted optima. See \Cref{appx:JumpStability} for the proof.
\begin{restatable}{lemma}{lemmaUnsortedUnstableMaxJumpOpt}
    \label{lemma:unsorted-unstable-max-jump-opt}
    For \textsc{Max}-\textsc{Jump-\{UIS,HIS\}} games, there is an instance with only unsorted optima.
\end{restatable}

Computing the optimal coalition structure in \textsc{Average-Jump} games is challenging even for $k=2$, due to the interdependence between the agents' values and their positions within the coalition structure. 
To gain a better understanding of these optima, we study their structural properties and prove several lemmata and observations (see \Cref{appx:avg-jump-opt}), which lead us to the following conjecture.
\begin{restatable}{conjecture}{avJumpOptConj}
\label{conj}
    In \textsc{Average-Jump} games with $k=2$, the optimal coalition structure is sorted.
\end{restatable}
This conjecture points to an important structural property of \textsc{Average-Jump} games. If it holds, the optimal coalition structure is also sorted for any $k$, since otherwise two overlapping coalitions $C'$ and $C''$ in $\mathcal{C}$ can be rearranged to reduce the total cost. 
In particular, the existence of sorted optimal structures hints at the possibility of a dynamic programming approach, similar to those used for one-dimensional clustering problems such as $k$-means and $k$-medians~\cite{GronlundFastExactKMeans2017} and related clustering models~\cite{FeldmanHedonicClusteringGames2015}.
However, it remains unclear whether such an approach can be directly applied or adapted to our setting, even if the conjecture holds, due to the dependencies between agents’ costs and coalition boundaries.

\section{Quality of Equilibria}\label{sec:eq-quality}
In this section, we look at the impact of selfish behavior on the social cost, by looking at the price of anarchy (PoA) and price of stability (PoS) for both \textsc{Swap}- and \textsc{Jump} games.

\subsection{Price of Anarchy}
\Cref{theorem:PoA} summarizes our results on the PoA, which is unbounded in most settings. We show results by providing worst-case instances for different game variants, split into \Cref{lemma:PoA-cutoff,lemma:PoA-Max-Avg-Jump-Swap,lemma:avg-jump-his-k2-poa,lemma:avg-jump-uis-k2-poa}.

\begin{theorem}\label{theorem:PoA}
    The PoA is unbounded for all games with all $k \geq 2$, except for \textsc{Avg-Jump-HIS} games with exactly two coalitions, where it is either 1 or at least $\Omega(n)$ depending on the agent values.
\end{theorem}

\begin{restatable}{lemma}{lemmaPoACutoff}\label{lemma:PoA-cutoff}
        The price of anarchy is unbounded for \textsc{Cutoff$_\lambda$-\{Swap, Jump\}} games with all $k \geq 2$.
\end{restatable}
\begin{proof}[Proof sketch]
    We give two different game instances for $k=2$ and $k > 2$ respectively, and show that they have an optimum $\mathcal{C}^*$ with social cost $0$ and admit a coalition structure $\mathcal{C}$ with positive social cost that is both swap and jump stable.

    For $k=2$, we assume $\lambda=1$ and consider a game with small $0.5\geq\varepsilon>0$ and the following two coalition structures
    \begin{align*}
        \mathcal{C}^* \coloneqq\;& \big\{\{0,1-\varepsilon,1\},\{1+\varepsilon,2\}\big\},\\
        \mathcal{C} \coloneqq\;& \big\{\{0,1,2\},\{1-\varepsilon,1+\varepsilon\}\big\}.
    \end{align*}
    
    For all $k>2$, consider the game with $n\coloneqq 4(k-1)$ agents and coalitions of size $k_1 \coloneqq 2(k-1)$ and $k_i = 2$ for all $i \in [k] \setminus \{1\}$. Then
    \begin{align*}
        \mathcal{C}^* \coloneqq\;& \big\{\{\overbrace{0\ldots0}^{2(k-1)}\},\,\{2,2\},\,\dots,\{k,k\}\big\},\; \text{and}\\
        \mathcal{C} \coloneqq\;& \big\{\overbrace{\{0,0\},\dots,\{0,0\}}^{k-1},\,\{2,2,\,3,3,\,\,\dots,\,k,k\}\big\}
    \end{align*}
    are the social optimum with cost $0$ and an equilibrium with positive cost, respectively.
\end{proof}

\begin{restatable}{lemma}{lemmaPoAMaxAvgJumpSwap}\label{lemma:PoA-Max-Avg-Jump-Swap}
    The price of anarchy of \textsc{\{Max, Average\}-\{Swap, Jump\}} games is unbounded for all $k \geq 3$ and $n\geq 2k+2$, independently of whether \textsc{UIS} or \textsc{HIS} is considered. For \textsc{Max-\{Swap, Jump\}} games, the number of coalitions can be reduced to $k=2$.
\end{restatable}
\begin{proof}[Proof sketch]
    We state the worst-case examples for both cost functions and refer to \Cref{appx:sub:PoA} for the general constructions. 
    
    For \textsc{Average}, let $L<M<H$ and $M \coloneqq \frac{H+L}{2}+1$, and consider a game with $k=3$ and $n=8$ admitting a coalition structure
        $$\big\{\{L,L\},\{L,L\},\{M,M,H,H\}\big\}.$$
    This is swap and jump stable for both the \textsc{Max} and the \textsc{Average} cost function. However, there is an optimum
        $$\big\{\{L,L,L,L\},\{M,M\},\{H,H\}\big\}.$$
    with a social cost of $0$. A similar construction works for the \textsc{Maximum} cost function even with $k=2$ and $n=8$. Here, 
        $$\big\{\{L,L,L,L\},\{H,H,H,H\}\big\}$$
    is the optimum with a social cost of $0$ while 
        $$\big\{\{L,L,H,H\},\{L,L,H,H\}\big\}$$
    is a swap and jump stable coalition structure with asymptotically maximal cost of $\Theta(n(H-L))$.
\end{proof}

\Cref{lemma:PoA-Max-Avg-Jump-Swap} only shows the PoA for \textsc{Avg} games with $k\geq 3$ coalitions. With \Cref{lemma:avg-jump-his-k2-poa} and \Cref{lemma:avg-jump-uis-k2-poa}, we take a closer look at $k=2$. The proofs can be found in \Cref{appx:sub:PoA}.

\begin{restatable}{lemma}{lemmaAvgJumpHISkTwoPoA}
\label{lemma:avg-jump-his-k2-poa}
    For \textsc{Avg-Jump-HIS} games, if $\SC(\OPT)=0$ for  $k=2$ coalitions, then the social cost of any equilibrium is also 0.
    If $\SC(\OPT)>0$, then the PoA is in $\Omega(n)$ and cannot be unbounded, even for instances with two coalitions and only 3 different agent values. 
\end{restatable}

Compared to \textsc{Average-Jump-HIS} games, where an unbounded PoA is impossible for $k=2$, we show that \textsc{Average-Jump-UIS} games have an unbounded PoA, since the grand coaliton is stable.

\begin{restatable}{lemma}{lemmaAvgJumpUISkTwoPoA}
\label{lemma:avg-jump-uis-k2-poa}
    For all three \textsc{Jump-UIS} games, even games with $k=2$ coalitions and only two different agent values have unbounded PoA.
\end{restatable}

\subsection{Price of Stability}
While the results on the PoA are similar for \textsc{Swap} and \textsc{Jump} games, the PoS differs between \textsc{Swap} and \textsc{Jump}. Although the optima of all three \textsc{Swap} games are always stable, this is not necessarily true for the respective \textsc{Jump} games.

An overview of the results is given by \Cref{theorem:PoS}. For that, we remind the reader of the definition of nice \textsc{Cutoff} games given at the end of \Cref{sec:model}.
\begin{theorem}\label{theorem:PoS}
    The PoS for all \textsc{Swap} games and nice \textsc{Cutoff$_\lambda$} games is 1.
    For all general \textsc{Jump} games, it is strictly larger than 1.
\end{theorem}

For \textsc{Cutoff$_\lambda$-Swap} and \textsc{Avg-Swap} games, this follows directly from the social cost being potential functions, while the proof for \textsc{Max-Swap} games is more technical. See \Cref{appx:sub:PoS} for both proofs.

\begin{restatable}{corollary}{corCutoffAvgSwapPos}\label{cor:cutoff-avg-swap-pos}
    The price of stability is exactly $1$ for all \textsc{Cutoff$_\lambda$-Swap} and \textsc{Avg-Swap} games.
\end{restatable}

\begin{restatable}{lemma}{lemmaMaxSwapPoS}\label{lemma:max-swap-pos}
    The price of stability is $1$ for \textsc{Max-Swap} games.
\end{restatable}

In contrast to \textsc{Swap} games, the PoS in \textsc{Jump} games is greater than~1 in general, as optima can be unstable. However, there is an interesting subclass of \textsc{Cutoff$_\lambda$-Jump} games with a PoS of $1$.

\begin{restatable}{lemma}{lemmaAvgMaxJumpPoS}\label{lemma:avg-max-jump-pos}
    For \textsc{Avg-} and \textsc{Max-Jump} games, it holds that PoS $>1$.
\end{restatable}
\begin{proof}[Proof sketch]
    Consider the following game instance with ${n=7}$ agents, $k=2$ coalitions and values allowing the coalition structure
        $$\mathcal{C}^* \coloneqq\big\{\{1,1,1\},\{4,6,8,8\}\big\}.$$
    In \Cref{appx:sub:PoS}, we prove for both cost functions that this coalition structure $\mathcal{C}^*$ indeed (1) is the social optimum and (2) has an improving jump for the agent with value $4$.
\end{proof}

\begin{restatable}{lemma}{lemmaPoSCutoffJump}\label{lemma:pos-cutoff-jump}
    The PoS of nice \textsc{Cutoff$_\lambda$-Jump} games is exactly $1$, while the PoS over all other \textsc{Cutoff$_\lambda$-Jump} games is greater than $1$.
\end{restatable}
\begin{proof}[Proof sketch]
    For nice \textsc{Cutoff$_\lambda$-Jump} games, the optimum has a social cost of $0$, since the agents can be partitioned into $k$ coalitions based on the unique $\lambda$-blocks by which they are covered. As coalition structures with social cost of $0$ are stable, the PoS is~$1$.

    For the second part, consider a \textsc{Cutoff$_\lambda$-Jump} game instance with $n=8$ agents, $k=2$ coalitions and coalition structure
        $$\mathcal{C}^* \coloneqq \big\{\{0, \frac14\varepsilon, \frac24\varepsilon, \frac24\varepsilon\},\; \{ \underline{\lambda + \frac14\varepsilon}, \lambda + \frac34\varepsilon, 2\lambda + \frac14\varepsilon, 2\lambda+\varepsilon\}\big\}$$
    for some $\varepsilon > 0$. In \Cref{appx:sub:PoS}, we show that this coalition structure $\mathcal{C}^*$ is an optimum and not jump stable, as the underlined agent has an improving jump.
\end{proof}

\section*{Conclusion}
We consider a hedonic games model, where each agent has a fixed value, and evaluates the cost of each coalition by accumulating the values of all other agents in its coalition, using three natural cost functions.
It is clear that in this setting, many more variants could be implemented, e.g. it would also be interesting to explore alternative cost functions, such as vectors of distances, squared distances, or $x$-nearest-neighbor formulations.
Within this setting, we study two stability concepts and show the existence of equilibria and their quality compared to the utilitarian social optimum. 
While unsorted optima exist for most variants, it is yet unclear whether they also exist for \textsc{Avg-Jump} games.

The structural property of being sorted, proposed in our conjecture, is key to understanding how optimal coalition structures behave in \textsc{Avg-Jump} games.
Confirming or disproving this conjecture would therefore be an important step toward obtaining a clearer picture of both equilibrium quality and social optimality, and could guide future algorithmic work.
Still, proving this property is far from easy, even for the simple case of $k=2$. 

Since there is a substantial gap between unbounded PoA and the constant lower bound on PoS, and because the worst-case instances appear somewhat artificial, future work could focus on random or real world instances, or on stability notions allowing more cooperation.
Regarding strong stability, a result by~\citet{MilchtaichStabilitySegregationGroup2002} suggests instances with no equilibria, so models permitting limited but not arbitrary cooperation could be of interest, especially as some of the reported high PoA instances (e.g. \Cref{lemma:PoA-Max-Avg-Jump-Swap}) are clearly not stable if multiple agents could coordinate.

Another promising direction would be to limit the size of coalitions, e.g. no coalition can be larger than $x\in\N$, to prevent the formation of the grand coalition in the \textsc{UIS} setting. Moreover, an unbounded PoA provides limited insight. In particular, for the utility variant of nice \textsc{Cutoff$_\lambda$} games, the price of anarchy becomes bounded. Therefore, analyzing the utility variant may offer additional understanding of the game. Other avenues for research include analyzing the core and the computational complexity of finding an equilibrium.

\begin{acks}

\texttt{This research was (partially) funded by the HPI Research School on Foundations of AI (FAI).}

\end{acks}
\pagebreak

\bibliographystyle{ACM-Reference-Format} 
\bibliography{refs}


\begin{thebibliography}{39}


\ifx \showCODEN    \undefined \def \showCODEN     #1{\unskip}     \fi
\ifx \showDOI      \undefined \def \showDOI       #1{#1}\fi
\ifx \showISBNx    \undefined \def \showISBNx     #1{\unskip}     \fi
\ifx \showISBNxiii \undefined \def \showISBNxiii  #1{\unskip}     \fi
\ifx \showISSN     \undefined \def \showISSN      #1{\unskip}     \fi
\ifx \showLCCN     \undefined \def \showLCCN      #1{\unskip}     \fi
\ifx \shownote     \undefined \def \shownote      #1{#1}          \fi
\ifx \showarticletitle \undefined \def \showarticletitle #1{#1}   \fi
\ifx \showURL      \undefined \def \showURL       {\relax}        \fi
\providecommand\bibfield[2]{#2}
\providecommand\bibinfo[2]{#2}
\providecommand\natexlab[1]{#1}
\providecommand\showeprint[2][]{arXiv:#2}

\bibitem[\protect\citeauthoryear{Aamand, Chen, Li, Silwal, Sukprasert,
  Vakilian, and Zhang}{Aamand et~al\mbox{.}}{2023}]%
        {Aamand2023}
\bibfield{author}{\bibinfo{person}{Anders Aamand}, \bibinfo{person}{Justin~Y.
  Chen}, \bibinfo{person}{Allen Li}, \bibinfo{person}{Sandeep Silwal},
  \bibinfo{person}{Pattara Sukprasert}, \bibinfo{person}{Ali Vakilian}, {and}
  \bibinfo{person}{Fred Zhang}.} \bibinfo{year}{2023}\natexlab{}.
\newblock \showarticletitle{Constant approximation for individual preference
  stable clustering}. In \bibinfo{booktitle}{\emph{NeuIPS 2023}}. Article
  \bibinfo{articleno}{1892}, \bibinfo{numpages}{16}~pages.
\newblock
\urldef\tempurl%
\url{http://papers.nips.cc/paper\_files/paper/2023/hash/881259965dacb9f42967aae84a157283-Abstract-Conference.html}
\showURL{%
\tempurl}


\bibitem[\protect\citeauthoryear{Agarwal, Elkind, Gan, Igarashi, Suksompong,
  and Voudouris}{Agarwal et~al\mbox{.}}{2021}]%
        {AgarwalSchellingGamesGraphs2021}
\bibfield{author}{\bibinfo{person}{Aishwarya Agarwal}, \bibinfo{person}{Edith
  Elkind}, \bibinfo{person}{Jiarui Gan}, \bibinfo{person}{Ayumi Igarashi},
  \bibinfo{person}{Warut Suksompong}, {and} \bibinfo{person}{Alexandros~A.
  Voudouris}.} \bibinfo{year}{2021}\natexlab{}.
\newblock \showarticletitle{Schelling Games on Graphs}.
\newblock \bibinfo{journal}{\emph{Artif. Intell.}}  \bibinfo{volume}{301}
  (\bibinfo{year}{2021}), \bibinfo{pages}{103576}.
\newblock
\urldef\tempurl%
\url{https://doi.org/10.1016/J.ARTINT.2021.103576}
\showDOI{\tempurl}


\bibitem[\protect\citeauthoryear{Ahmadi, Awasthi, Khuller, Kleindessner,
  Morgenstern, Sukprasert, and Vakilian}{Ahmadi et~al\mbox{.}}{2022}]%
        {Ahmadi2022}
\bibfield{author}{\bibinfo{person}{Saba Ahmadi}, \bibinfo{person}{Pranjal
  Awasthi}, \bibinfo{person}{Samir Khuller}, \bibinfo{person}{Matth{\"{a}}us
  Kleindessner}, \bibinfo{person}{Jamie Morgenstern}, \bibinfo{person}{Pattara
  Sukprasert}, {and} \bibinfo{person}{Ali Vakilian}.}
  \bibinfo{year}{2022}\natexlab{}.
\newblock \showarticletitle{Individual Preference Stability for Clustering}. In
  \bibinfo{booktitle}{\emph{{ICML} 2022}}. \bibinfo{pages}{197--246}.
\newblock
\urldef\tempurl%
\url{https://proceedings.mlr.press/v162/ahmadi22a.html}
\showURL{%
\tempurl}


\bibitem[\protect\citeauthoryear{Alcalde}{Alcalde}{1994}]%
        {AlcaldeExchangeproofnessDivorceproofnessStability1994}
\bibfield{author}{\bibinfo{person}{Jos{\'e} Alcalde}.}
  \bibinfo{year}{1994}\natexlab{}.
\newblock \showarticletitle{Exchange-Proofness or Divorce-Proofness?
  {{Stability}} in One-Sided Matching Markets}.
\newblock \bibinfo{journal}{\emph{Economic Design}} \bibinfo{volume}{1},
  \bibinfo{number}{1} (\bibinfo{year}{1994}).
\newblock
\urldef\tempurl%
\url{https://doi.org/10.1007/BF02716626}
\showDOI{\tempurl}


\bibitem[\protect\citeauthoryear{Anshelevich, Dasgupta, Kleinberg, Tardos,
  Wexler, and Roughgarden}{Anshelevich et~al\mbox{.}}{2008}]%
        {AnshelevichPriceStabilityNetwork2008}
\bibfield{author}{\bibinfo{person}{Elliot Anshelevich},
  \bibinfo{person}{Anirban Dasgupta}, \bibinfo{person}{Jon~M. Kleinberg},
  \bibinfo{person}{{\'E}va Tardos}, \bibinfo{person}{Tom Wexler}, {and}
  \bibinfo{person}{Tim Roughgarden}.} \bibinfo{year}{2008}\natexlab{}.
\newblock \showarticletitle{The Price of Stability for Network Design with Fair
  Cost Allocation}.
\newblock \bibinfo{journal}{\emph{SIAM J. Comput.}} \bibinfo{volume}{38},
  \bibinfo{number}{4} (\bibinfo{year}{2008}), \bibinfo{pages}{1602--1623}.
\newblock
\urldef\tempurl%
\url{https://doi.org/10.1137/070680096}
\showDOI{\tempurl}


\bibitem[\protect\citeauthoryear{Aziz, Brandl, Brandt, Harrenstein, Olsen, and
  Peters}{Aziz et~al\mbox{.}}{2019}]%
        {AzizFractionalHedonicGames2019}
\bibfield{author}{\bibinfo{person}{Haris Aziz}, \bibinfo{person}{Florian
  Brandl}, \bibinfo{person}{Felix Brandt}, \bibinfo{person}{Paul Harrenstein},
  \bibinfo{person}{Martin Olsen}, {and} \bibinfo{person}{Dominik Peters}.}
  \bibinfo{year}{2019}\natexlab{}.
\newblock \showarticletitle{Fractional {{Hedonic Games}}}.
\newblock \bibinfo{journal}{\emph{ACM Trans. Econ. Comput.}}
  \bibinfo{volume}{7}, \bibinfo{number}{2} (\bibinfo{date}{May}
  \bibinfo{year}{2019}), \bibinfo{pages}{1--29}.
\newblock
\showISSN{2167-8375, 2167-8383}
\urldef\tempurl%
\url{https://doi.org/10.1145/3327970}
\showDOI{\tempurl}


\bibitem[\protect\citeauthoryear{Aziz and Goldwaser}{Aziz and
  Goldwaser}{2017}]%
        {AzizCoalitionalExchangeStable2017}
\bibfield{author}{\bibinfo{person}{Haris Aziz} {and} \bibinfo{person}{Adrian
  Goldwaser}.} \bibinfo{year}{2017}\natexlab{}.
\newblock \showarticletitle{Coalitional {{Exchange Stable Matchings}} in
  {{Marriage}} and {{Roommate Markets}}}. In
  \bibinfo{booktitle}{\emph{{{AAMAS}} 2017}}. \bibinfo{pages}{1475--1477}.
\newblock
\urldef\tempurl%
\url{https://dl.acm.org/doi/10.5555/3091125.3091334}
\showURL{%
\tempurl}


\bibitem[\protect\citeauthoryear{Aziz, Harrenstein, and Pyrga}{Aziz
  et~al\mbox{.}}{2012}]%
        {AzizIndividualbasedStabilityHedonic2012}
\bibfield{author}{\bibinfo{person}{Haris Aziz}, \bibinfo{person}{Paul
  Harrenstein}, {and} \bibinfo{person}{Evangelia Pyrga}.}
  \bibinfo{year}{2012}\natexlab{}.
\newblock \showarticletitle{Individual-Based Stability in Hedonic Games
  Depending on the Best or Worst Players}. In
  \bibinfo{booktitle}{\emph{{{AAMAS}} 2012}}. \bibinfo{pages}{1311--1312}.
\newblock
\urldef\tempurl%
\url{https://dl.acm.org/doi/10.5555/2343896.2343979}
\showURL{%
\tempurl}


\bibitem[\protect\citeauthoryear{Aziz and Savani}{Aziz and Savani}{2016}]%
        {AzizHedonicGames2016}
\bibfield{author}{\bibinfo{person}{Haris Aziz} {and} \bibinfo{person}{Rahul
  Savani}.} \bibinfo{year}{2016}\natexlab{}.
\newblock \showarticletitle{Hedonic Games}.
\newblock In \bibinfo{booktitle}{\emph{Handbook of {{Computational Social
  Choice}}}}. \bibinfo{pages}{356--376}.
\newblock
\showISBNx{1-107-06043-5}
\urldef\tempurl%
\url{https://doi.org/10.1017/cbo9781107446984.016}
\showDOI{\tempurl}


\bibitem[\protect\citeauthoryear{Banerjee, Konishi, and S{\"o}nmez}{Banerjee
  et~al\mbox{.}}{2001}]%
        {BanerjeeCoreSimpleCoalition2001}
\bibfield{author}{\bibinfo{person}{Suryapratim Banerjee},
  \bibinfo{person}{Hideo Konishi}, {and} \bibinfo{person}{Tayfun S{\"o}nmez}.}
  \bibinfo{year}{2001}\natexlab{}.
\newblock \showarticletitle{Core in a Simple Coalition Formation Game}.
\newblock \bibinfo{journal}{\emph{Soc. Choice Welf.}} \bibinfo{volume}{18},
  \bibinfo{number}{1} (\bibinfo{year}{2001}), \bibinfo{pages}{135--153}.
\newblock
\urldef\tempurl%
\url{https://doi.org/10.1007/S003550000067}
\showDOI{\tempurl}


\bibitem[\protect\citeauthoryear{Bil{\`o}, Bil{\`o}, D{\"o}ring, Lenzner,
  Molitor, and Schmidt}{Bil{\`o} et~al\mbox{.}}{2023}]%
        {BiloSchellingGamesContinuous2023}
\bibfield{author}{\bibinfo{person}{Davide Bil{\`o}}, \bibinfo{person}{Vittorio
  Bil{\`o}}, \bibinfo{person}{Michelle D{\"o}ring}, \bibinfo{person}{Pascal
  Lenzner}, \bibinfo{person}{Louise Molitor}, {and} \bibinfo{person}{Jonas
  Schmidt}.} \bibinfo{year}{2023}\natexlab{}.
\newblock \showarticletitle{Schelling {{Games}} with {{Continuous Types}}}. In
  \bibinfo{booktitle}{\emph{{{IJCAI}} 23}}. \bibinfo{pages}{2520--2527}.
\newblock
\showISBNx{978-1-956792-03-4}
\urldef\tempurl%
\url{https://doi.org/10.24963/ijcai.2023/280}
\showDOI{\tempurl}


\bibitem[\protect\citeauthoryear{Bil{\`o}, Bil{\`o}, Lenzner, and
  Molitor}{Bil{\`o} et~al\mbox{.}}{2022a}]%
        {BiloTopologicalInfluenceLocality2022}
\bibfield{author}{\bibinfo{person}{Davide Bil{\`o}}, \bibinfo{person}{Vittorio
  Bil{\`o}}, \bibinfo{person}{Pascal Lenzner}, {and} \bibinfo{person}{Louise
  Molitor}.} \bibinfo{year}{2022}\natexlab{a}.
\newblock \showarticletitle{Topological Influence and Locality in Swap
  Schelling Games}.
\newblock \bibinfo{journal}{\emph{Auton. Agents Multi Agent Syst.}}
  \bibinfo{volume}{36}, \bibinfo{number}{2} (\bibinfo{year}{2022}),
  \bibinfo{pages}{47}.
\newblock
\urldef\tempurl%
\url{https://doi.org/10.1007/S10458-022-09573-7}
\showDOI{\tempurl}


\bibitem[\protect\citeauthoryear{Bil{\`o}, Monaco, and Moscardelli}{Bil{\`o}
  et~al\mbox{.}}{2022b}]%
        {BiloHedonicGamesFixedsize2022}
\bibfield{author}{\bibinfo{person}{Vittorio Bil{\`o}},
  \bibinfo{person}{Gianpiero Monaco}, {and} \bibinfo{person}{Luca
  Moscardelli}.} \bibinfo{year}{2022}\natexlab{b}.
\newblock \showarticletitle{Hedonic Games with Fixed-Size Coalitions}. In
  \bibinfo{booktitle}{\emph{{{AAAI}} 2022}}. \bibinfo{pages}{9287--9295}.
\newblock
\urldef\tempurl%
\url{https://doi.org/10.1609/AAAI.V36I9.21156}
\showDOI{\tempurl}


\bibitem[\protect\citeauthoryear{Bogomolnaia and Jackson}{Bogomolnaia and
  Jackson}{2002}]%
        {BogomolnaiaStabilityHedonicCoalition2002}
\bibfield{author}{\bibinfo{person}{Anna Bogomolnaia} {and}
  \bibinfo{person}{Matthew~O. Jackson}.} \bibinfo{year}{2002}\natexlab{}.
\newblock \showarticletitle{The Stability of Hedonic Coalition Structures}.
\newblock \bibinfo{journal}{\emph{Games Econ. Behav.}} \bibinfo{volume}{38},
  \bibinfo{number}{2} (\bibinfo{year}{2002}), \bibinfo{pages}{201--230}.
\newblock
\urldef\tempurl%
\url{https://doi.org/10.1006/game.2001.0877}
\showDOI{\tempurl}


\bibitem[\protect\citeauthoryear{Br{\^a}nzei and Larson}{Br{\^a}nzei and
  Larson}{2011}]%
        {BranzeiSocialDistanceGames2011}
\bibfield{author}{\bibinfo{person}{Simina Br{\^a}nzei} {and}
  \bibinfo{person}{Kate Larson}.} \bibinfo{year}{2011}\natexlab{}.
\newblock \showarticletitle{Social Distance Games}. In
  \bibinfo{booktitle}{\emph{{{IJCAI}} 2011}}. \bibinfo{pages}{91--96}.
\newblock
\urldef\tempurl%
\url{https://doi.org/10.5591/978-1-57735-516-8/IJCAI11-027}
\showDOI{\tempurl}


\bibitem[\protect\citeauthoryear{Cechl{\'a}rov{\'a} and
  Hajdukov{\'a}}{Cechl{\'a}rov{\'a} and Hajdukov{\'a}}{2004}]%
        {CechlarovaStablePartitionsWpreferences2004}
\bibfield{author}{\bibinfo{person}{Katar{\'\i}na Cechl{\'a}rov{\'a}} {and}
  \bibinfo{person}{Jana Hajdukov{\'a}}.} \bibinfo{year}{2004}\natexlab{}.
\newblock \showarticletitle{Stable Partitions with {{W-preferences}}}.
\newblock \bibinfo{journal}{\emph{Discret. Appl. Math.}} \bibinfo{volume}{138},
  \bibinfo{number}{3} (\bibinfo{year}{2004}), \bibinfo{pages}{333--347}.
\newblock
\urldef\tempurl%
\url{https://doi.org/10.1016/S0166-218X(03)00464-5}
\showDOI{\tempurl}


\bibitem[\protect\citeauthoryear{Cechl{\'a}rov{\'a} and
  {Romero-Medina}}{Cechl{\'a}rov{\'a} and {Romero-Medina}}{2001}]%
        {CechlarovaStabilityCoalitionFormation2001}
\bibfield{author}{\bibinfo{person}{Katar{\'\i}na Cechl{\'a}rov{\'a}} {and}
  \bibinfo{person}{Antonio {Romero-Medina}}.} \bibinfo{year}{2001}\natexlab{}.
\newblock \showarticletitle{Stability in Coalition Formation Games}.
\newblock \bibinfo{journal}{\emph{Int. J. Game Theory}} \bibinfo{volume}{29},
  \bibinfo{number}{4} (\bibinfo{date}{May} \bibinfo{year}{2001}),
  \bibinfo{pages}{487--494}.
\newblock
\showISSN{0020-7276, 1432-1270}
\urldef\tempurl%
\url{https://doi.org/10.1007/s001820000053}
\showDOI{\tempurl}


\bibitem[\protect\citeauthoryear{Chauhan, Lenzner, and Molitor}{Chauhan
  et~al\mbox{.}}{2018}]%
        {ChauhanSchellingSegregationStrategic2018}
\bibfield{author}{\bibinfo{person}{Ankit Chauhan}, \bibinfo{person}{Pascal
  Lenzner}, {and} \bibinfo{person}{Louise Molitor}.}
  \bibinfo{year}{2018}\natexlab{}.
\newblock \showarticletitle{Schelling Segregation with Strategic Agents}. In
  \bibinfo{booktitle}{\emph{{{SAGT}} 2018}}, Vol.~\bibinfo{volume}{11059}.
  \bibinfo{pages}{137--149}.
\newblock
\urldef\tempurl%
\url{https://doi.org/10.1007/978-3-319-99660-8\_13}
\showDOI{\tempurl}


\bibitem[\protect\citeauthoryear{Cseh, Fleiner, and Harj{\'a}n}{Cseh
  et~al\mbox{.}}{2019}]%
        {CsehParetoOptimalCoalitions2019}
\bibfield{author}{\bibinfo{person}{{\'A}gnes Cseh}, \bibinfo{person}{Tam{\'a}s
  Fleiner}, {and} \bibinfo{person}{Petra Harj{\'a}n}.}
  \bibinfo{year}{2019}\natexlab{}.
\newblock \showarticletitle{Pareto {{Optimal Coalitions}} of {{Fixed Size}}}.
\newblock \bibinfo{journal}{\emph{jMID}} \bibinfo{volume}{4},
  \bibinfo{number}{1} (\bibinfo{year}{2019}), \bibinfo{pages}{87--108}.
\newblock
\showISSN{2399844X, 23998458}
\urldef\tempurl%
\url{https://doi.org/10.22574/jmid.2019.11.003}
\showDOI{\tempurl}


\bibitem[\protect\citeauthoryear{Damamme, Beynier, Chevaleyre, and
  Maudet}{Damamme et~al\mbox{.}}{2015}]%
        {DamammePowerSwapDeals2015}
\bibfield{author}{\bibinfo{person}{Anastasia Damamme},
  \bibinfo{person}{Aur{\'e}lie Beynier}, \bibinfo{person}{Yann Chevaleyre},
  {and} \bibinfo{person}{Nicolas Maudet}.} \bibinfo{year}{2015}\natexlab{}.
\newblock \showarticletitle{The Power of Swap Deals in Distributed Resource
  Allocation}. In \bibinfo{booktitle}{\emph{{{AAMAS}} 2015}}.
  \bibinfo{pages}{625--633}.
\newblock
\showISBNx{978-1-4503-3413-6}
\urldef\tempurl%
\url{https://dl.acm.org/doi/10.5555/2772879.2773235}
\showURL{%
\tempurl}


\bibitem[\protect\citeauthoryear{Dimitrov, Borm, Hendrickx, and Sung}{Dimitrov
  et~al\mbox{.}}{2006}]%
        {DimitrovSimplePrioritiesCore2006}
\bibfield{author}{\bibinfo{person}{Dinko Dimitrov}, \bibinfo{person}{Peter
  Borm}, \bibinfo{person}{Ruud Hendrickx}, {and} \bibinfo{person}{Shao~Chin
  Sung}.} \bibinfo{year}{2006}\natexlab{}.
\newblock \showarticletitle{Simple Priorities and Core Stability in Hedonic
  Games}.
\newblock \bibinfo{journal}{\emph{Soc. Choice Welf.}} \bibinfo{volume}{26},
  \bibinfo{number}{2} (\bibinfo{year}{2006}), \bibinfo{pages}{421--433}.
\newblock
\urldef\tempurl%
\url{https://doi.org/10.1007/S00355-006-0104-4}
\showDOI{\tempurl}


\bibitem[\protect\citeauthoryear{Dr{\`e}ze and Greenberg}{Dr{\`e}ze and
  Greenberg}{1980}]%
        {DrezeHedonicCoalitionsOptimality1980}
\bibfield{author}{\bibinfo{person}{Jacques~H Dr{\`e}ze} {and}
  \bibinfo{person}{Joseph Greenberg}.} \bibinfo{year}{1980}\natexlab{}.
\newblock \showarticletitle{Hedonic Coalitions: {{Optimality}} and
  {{Stability}}}.
\newblock \bibinfo{journal}{\emph{Econometrica}} \bibinfo{volume}{48},
  \bibinfo{number}{4} (\bibinfo{year}{1980}), \bibinfo{pages}{987--1003}.
\newblock
\showISSN{00129682, 14680262}
\urldef\tempurl%
\url{https://doi.org/10.2307/1912943}
\showDOI{\tempurl}
\showeprint[jstor]{1912943}


\bibitem[\protect\citeauthoryear{Echzell, Friedrich, Lenzner, Molitor, Pappik,
  Sch{\"o}ne, Sommer, and Stangl}{Echzell et~al\mbox{.}}{2019}]%
        {EchzellConvergenceHardnessStrategic2019}
\bibfield{author}{\bibinfo{person}{Hagen Echzell}, \bibinfo{person}{Tobias
  Friedrich}, \bibinfo{person}{Pascal Lenzner}, \bibinfo{person}{Louise
  Molitor}, \bibinfo{person}{Marcus Pappik}, \bibinfo{person}{Friedrich
  Sch{\"o}ne}, \bibinfo{person}{Fabian Sommer}, {and} \bibinfo{person}{David
  Stangl}.} \bibinfo{year}{2019}\natexlab{}.
\newblock \showarticletitle{Convergence and Hardness of Strategic Schelling
  Segregation}. In \bibinfo{booktitle}{\emph{{{WINE}} 2019}},
  Vol.~\bibinfo{volume}{11920}. \bibinfo{pages}{156--170}.
\newblock
\urldef\tempurl%
\url{https://doi.org/10.1007/978-3-030-35389-6\_12}
\showDOI{\tempurl}


\bibitem[\protect\citeauthoryear{Feldman, {Lewin-Eytan}, and Naor}{Feldman
  et~al\mbox{.}}{2015}]%
        {FeldmanHedonicClusteringGames2015}
\bibfield{author}{\bibinfo{person}{Moran Feldman}, \bibinfo{person}{Liane
  {Lewin-Eytan}}, {and} \bibinfo{person}{Joseph~(Seffi) Naor}.}
  \bibinfo{year}{2015}\natexlab{}.
\newblock \showarticletitle{Hedonic {{Clustering Games}}}.
\newblock \bibinfo{journal}{\emph{ACM Trans. Parallel Comput.}}
  \bibinfo{volume}{2}, \bibinfo{number}{1} (\bibinfo{date}{May}
  \bibinfo{year}{2015}), \bibinfo{pages}{1--48}.
\newblock
\showISSN{2329-4949, 2329-4957}
\urldef\tempurl%
\url{https://doi.org/10.1145/2742345}
\showDOI{\tempurl}


\bibitem[\protect\citeauthoryear{Flammini, Kodric, Olsen, and
  Varricchio}{Flammini et~al\mbox{.}}{2021}]%
        {FlamminiDistanceHedonicGames2021}
\bibfield{author}{\bibinfo{person}{Michele Flammini}, \bibinfo{person}{Bojana
  Kodric}, \bibinfo{person}{Martin Olsen}, {and} \bibinfo{person}{Giovanna
  Varricchio}.} \bibinfo{year}{2021}\natexlab{}.
\newblock \showarticletitle{Distance {{Hedonic Games}}}.
\newblock In \bibinfo{booktitle}{\emph{{{SOFSEM}} 2021: {{Theory}} and
  {{Practice}} of {{Computer Science}}}}. Vol.~\bibinfo{volume}{12607}.
  \bibinfo{pages}{159--174}.
\newblock
\showISBNx{978-3-030-67730-5 978-3-030-67731-2}
\urldef\tempurl%
\url{https://doi.org/10.1007/978-3-030-67731-2_12}
\showDOI{\tempurl}


\bibitem[\protect\citeauthoryear{Gr{\o}nlund, Larsen, Mathiasen, and
  Nielsen}{Gr{\o}nlund et~al\mbox{.}}{2017}]%
        {GronlundFastExactKMeans2017}
\bibfield{author}{\bibinfo{person}{Allan Gr{\o}nlund},
  \bibinfo{person}{Kasper~Green Larsen}, \bibinfo{person}{Alexander Mathiasen},
  {and} \bibinfo{person}{Jesper~Sindahl Nielsen}.}
  \bibinfo{year}{2017}\natexlab{}.
\newblock \showarticletitle{Fast Exact K-{{Means}}, k-{{Medians}} and Bregman
  Divergence Clustering in {{1D}}}.
\newblock \bibinfo{journal}{\emph{CoRR}} (\bibinfo{year}{2017}).
\newblock
\urldef\tempurl%
\url{https://doi.org/10.48550/arXiv.1701.07204}
\showDOI{\tempurl}


\bibitem[\protect\citeauthoryear{Hajdukov{\'a}}{Hajdukov{\'a}}{2006}]%
        {HajdukovaCoalitionFormationGames2006}
\bibfield{author}{\bibinfo{person}{Jana Hajdukov{\'a}}.}
  \bibinfo{year}{2006}\natexlab{}.
\newblock \showarticletitle{Coalition Formation Games: A Survey}.
\newblock \bibinfo{journal}{\emph{IGTR}} \bibinfo{volume}{8},
  \bibinfo{number}{4} (\bibinfo{year}{2006}), \bibinfo{pages}{613--641}.
\newblock
\urldef\tempurl%
\url{https://doi.org/10.1142/S0219198906001144}
\showDOI{\tempurl}


\bibitem[\protect\citeauthoryear{Kanellopoulos, Kyropoulou, and
  Voudouris}{Kanellopoulos et~al\mbox{.}}{2021}]%
        {KanellopoulosModifiedSchellingGames2021}
\bibfield{author}{\bibinfo{person}{Panagiotis Kanellopoulos},
  \bibinfo{person}{Maria Kyropoulou}, {and} \bibinfo{person}{Alexandros~A.
  Voudouris}.} \bibinfo{year}{2021}\natexlab{}.
\newblock \showarticletitle{Modified Schelling Games}.
\newblock \bibinfo{journal}{\emph{Theor. Comput. Sci.}}  \bibinfo{volume}{880}
  (\bibinfo{year}{2021}), \bibinfo{pages}{1--19}.
\newblock
\urldef\tempurl%
\url{https://doi.org/10.1016/J.TCS.2021.05.032}
\showDOI{\tempurl}


\bibitem[\protect\citeauthoryear{Kerkmann, Lang, Rey, Rothe, Schadrack, and
  Schend}{Kerkmann et~al\mbox{.}}{2020}]%
        {KerkmannHedonicGamesOrdinal2020}
\bibfield{author}{\bibinfo{person}{Anna~Maria Kerkmann},
  \bibinfo{person}{J{\'e}r{\^o}me Lang}, \bibinfo{person}{Anja Rey},
  \bibinfo{person}{J{\"o}rg Rothe}, \bibinfo{person}{Hilmar Schadrack}, {and}
  \bibinfo{person}{Lena Schend}.} \bibinfo{year}{2020}\natexlab{}.
\newblock \showarticletitle{Hedonic {{Games}} with {{Ordinal Preferences}} and
  {{Thresholds}}}.
\newblock \bibinfo{journal}{\emph{J. Artif. Intell. Res.}}
  \bibinfo{volume}{67} (\bibinfo{date}{April} \bibinfo{year}{2020}),
  \bibinfo{pages}{705--756}.
\newblock
\showISSN{1076-9757}
\urldef\tempurl%
\url{https://doi.org/10.1613/jair.1.11531}
\showDOI{\tempurl}


\bibitem[\protect\citeauthoryear{Koutsoupias and Papadimitriou}{Koutsoupias and
  Papadimitriou}{2009}]%
        {KoutsoupiasWorstcaseEquilibria2009}
\bibfield{author}{\bibinfo{person}{Elias Koutsoupias} {and}
  \bibinfo{person}{Christos~H. Papadimitriou}.}
  \bibinfo{year}{2009}\natexlab{}.
\newblock \showarticletitle{Worst-Case Equilibria}.
\newblock \bibinfo{journal}{\emph{Comput. Sci. Rev.}} \bibinfo{volume}{3},
  \bibinfo{number}{2} (\bibinfo{year}{2009}), \bibinfo{pages}{65--69}.
\newblock
\urldef\tempurl%
\url{https://doi.org/10.1016/J.COSREV.2009.04.003}
\showDOI{\tempurl}


\bibitem[\protect\citeauthoryear{Lang, Rey, Rothe, Schadrack, and Schend}{Lang
  et~al\mbox{.}}{2015}]%
        {LangRepresentingSolvingHedonic2015}
\bibfield{author}{\bibinfo{person}{J{\'e}r{\^o}me Lang}, \bibinfo{person}{Anja
  Rey}, \bibinfo{person}{J{\"o}rg Rothe}, \bibinfo{person}{Hilmar Schadrack},
  {and} \bibinfo{person}{Lena Schend}.} \bibinfo{year}{2015}\natexlab{}.
\newblock \showarticletitle{Representing and Solving Hedonic Games with Ordinal
  Preferences and Thresholds}. In \bibinfo{booktitle}{\emph{{{AAMAS}} 2015}}.
  \bibinfo{pages}{1229--1237}.
\newblock
\showISBNx{978-1-4503-3413-6}
\urldef\tempurl%
\url{https://dl.acm.org/doi/10.5555/2772879.2773308}
\showURL{%
\tempurl}


\bibitem[\protect\citeauthoryear{Li}{Li}{2021}]%
        {LiFractionalHedonicGames2021}
\bibfield{author}{\bibinfo{person}{Fu Li}.} \bibinfo{year}{2021}\natexlab{}.
\newblock \showarticletitle{Fractional {{Hedonic Games With}} a {{Limited
  Number}} of {{Coalitions}}}. In \bibinfo{booktitle}{\emph{{{ICTCS}} 2021}},
  Vol.~\bibinfo{volume}{3072}. \bibinfo{pages}{205--218}.
\newblock
\urldef\tempurl%
\url{https://ceur-ws.org/Vol-3072/paper18.pdf}
\showURL{%
\tempurl}


\bibitem[\protect\citeauthoryear{Milchtaich and Winter}{Milchtaich and
  Winter}{2002}]%
        {MilchtaichStabilitySegregationGroup2002}
\bibfield{author}{\bibinfo{person}{Igal Milchtaich} {and} \bibinfo{person}{Eyal
  Winter}.} \bibinfo{year}{2002}\natexlab{}.
\newblock \showarticletitle{Stability and {{Segregation}} in {{Group
  Formation}}}.
\newblock \bibinfo{journal}{\emph{Games Econ. Behav.}} \bibinfo{volume}{38},
  \bibinfo{number}{2} (\bibinfo{date}{Feb.} \bibinfo{year}{2002}),
  \bibinfo{pages}{318--346}.
\newblock
\showISSN{08998256}
\urldef\tempurl%
\url{https://doi.org/10.1006/game.2001.0878}
\showDOI{\tempurl}


\bibitem[\protect\citeauthoryear{Miller}{Miller}{1980}]%
        {MillerCoalitionFormationCharacteristic1980}
\bibfield{author}{\bibinfo{person}{Charles~E Miller}.}
  \bibinfo{year}{1980}\natexlab{}.
\newblock \showarticletitle{Coalition Formation in Characteristic Function
  Games: {{Competitive}} Tests of Three Theories}.
\newblock \bibinfo{journal}{\emph{J. Exp. Soc. Psychol.}} \bibinfo{volume}{16},
  \bibinfo{number}{1} (\bibinfo{year}{1980}), \bibinfo{pages}{61--76}.
\newblock
\showISSN{0022-1031}
\urldef\tempurl%
\url{https://doi.org/10.1016/0022-1031(80)90036-0}
\showDOI{\tempurl}


\bibitem[\protect\citeauthoryear{Monaco, Moscardelli, and Velaj}{Monaco
  et~al\mbox{.}}{2020}]%
        {MonacoStableOutcomesModified2020}
\bibfield{author}{\bibinfo{person}{Gianpiero Monaco}, \bibinfo{person}{Luca
  Moscardelli}, {and} \bibinfo{person}{Yllka Velaj}.}
  \bibinfo{year}{2020}\natexlab{}.
\newblock \showarticletitle{Stable Outcomes in Modified Fractional Hedonic
  Games}.
\newblock \bibinfo{journal}{\emph{Auton. Agents Multi Agent Syst.}}
  \bibinfo{volume}{34}, \bibinfo{number}{1} (\bibinfo{date}{April}
  \bibinfo{year}{2020}), \bibinfo{pages}{4}.
\newblock
\showISSN{1387-2532, 1573-7454}
\urldef\tempurl%
\url{https://doi.org/10.1007/s10458-019-09431-z}
\showDOI{\tempurl}


\bibitem[\protect\citeauthoryear{Monderer and Shapley}{Monderer and
  Shapley}{1996}]%
        {MondererPotentialGames1996}
\bibfield{author}{\bibinfo{person}{Dov Monderer} {and}
  \bibinfo{person}{Lloyd~S. Shapley}.} \bibinfo{year}{1996}\natexlab{}.
\newblock \showarticletitle{Potential {{Games}}}.
\newblock \bibinfo{journal}{\emph{Games Econ. Behav.}} \bibinfo{volume}{14},
  \bibinfo{number}{1} (\bibinfo{year}{1996}), \bibinfo{pages}{124--143}.
\newblock
\showISSN{0899-8256}
\urldef\tempurl%
\url{https://doi.org/10.1006/game.1996.0044}
\showDOI{\tempurl}


\bibitem[\protect\citeauthoryear{Rey and Rey}{Rey and Rey}{2022}]%
        {ReyFENHedonicGamesDistanceBased2022}
\bibfield{author}{\bibinfo{person}{Anja Rey} {and} \bibinfo{person}{Lisa Rey}.}
  \bibinfo{year}{2022}\natexlab{}.
\newblock \showarticletitle{{{FEN-Hedonic Games}} with {{Distance-Based
  Preferences}}}.
\newblock \bibinfo{journal}{\emph{CoRR}} (\bibinfo{date}{Jan.}
  \bibinfo{year}{2022}).
\newblock
\urldef\tempurl%
\url{https://doi.org/10.48550/arXiv.2201.13158}
\showDOI{\tempurl}


\bibitem[\protect\citeauthoryear{Schlueter and Yokoo}{Schlueter and
  Yokoo}{2024}]%
        {SchlueterFriendEnemyorientedHedonic2024}
\bibfield{author}{\bibinfo{person}{{\relax TJ} Schlueter} {and}
  \bibinfo{person}{Makoto Yokoo}.} \bibinfo{year}{2024}\natexlab{}.
\newblock \showarticletitle{Friend- and Enemy-Oriented Hedonic Games with
  Strangers}. In \bibinfo{booktitle}{\emph{{{PRIMA}} 2024}},
  Vol.~\bibinfo{volume}{15395}. \bibinfo{pages}{421--427}.
\newblock
\urldef\tempurl%
\url{https://doi.org/10.1007/978-3-031-77367-9\_31}
\showDOI{\tempurl}


\bibitem[\protect\citeauthoryear{Sless, Hazon, Kraus, and Wooldridge}{Sless
  et~al\mbox{.}}{2018}]%
        {SlessFormingCoalitionsFacilitating2018}
\bibfield{author}{\bibinfo{person}{Liat Sless}, \bibinfo{person}{Noam Hazon},
  \bibinfo{person}{Sarit Kraus}, {and} \bibinfo{person}{Michael Wooldridge}.}
  \bibinfo{year}{2018}\natexlab{}.
\newblock \showarticletitle{Forming k Coalitions and Facilitating Relationships
  in Social Networks}.
\newblock \bibinfo{journal}{\emph{Artif. Intell.}}  \bibinfo{volume}{259}
  (\bibinfo{date}{June} \bibinfo{year}{2018}), \bibinfo{pages}{217--245}.
\newblock
\showISSN{00043702}
\urldef\tempurl%
\url{https://doi.org/10.1016/j.artint.2018.03.004}
\showDOI{\tempurl}


\end{thebibliography}

\clearpage
\appendix 

\section{Omitted Details from Section 1}\label{appx:intro}
\begin{example}[\Cref{exmp:intro} continued]
    We look at the different cost functions separately. For that, let $x$ be the agent with value $7$ in coalition $C_2$. 
    
    With the \textsc{Average} cost function, agent $x$ has a cost of $\frac{|11-7|}{|C_2|-1} = 4$ in its current coalition $C_2$, but a cost of $\frac{3+3+0+0+1}{|C_1|}=1.4$ in coalition $C_1$ and a cost of $\frac{2+2+2}{|C_3|}=2$ in coalition $C_3$. Agent $x$ would thus prefer to jump to coalition $C_1$. Note that the \enquote{$-1$} is only used for the coalition the agent is part of.

    With the \textsc{Maximum} cost function, agent $x$ has a cost of $\max\{|11-7|\} = 4$ in its current coalition $C_2$, but a cost of $\max\{3,3,0,0,1\}=3$ in coalition $C_1$ and a cost of $\max\{2,2,2\}=2$ in coalition $C_3$. Agent $x$ would thus prefer to jump to coalition $C_3$. Note that we do not consider the distance from agent $x$ to itself in its current coalition.

    With the \textsc{Cutoff} cost function, we consider the thresholds $\lambda = 1$ and $\lambda = 2$. The other agent in the current coalition of agent $x$ has a value of $11$ and is thus outside the threshold radius for both $\lambda = 1$ and $\lambda = 2$. Therefore, agent $x$ has a maximum cost of $\frac{1}{|C_2|-1}=1$ in its current coalition $C_2$ in both cases. When assuming $\lambda = 1$, agent $x$ has two (three) agents outside the threshold in coalition $C_1$ ($C_2$). Therefore, the cost of agent $x$ is $\frac{2}{|C_1|}=0.4$ in coalition $C_1$ and $\frac{3}{|C_3|}=1$ in coalition $C_3$, implying that agent $x$ prefers to jump to coalition $C_1$ with $\lambda=1$. For $\lambda=2$, there are suddenly all agents in coalition $C_3$ inside the threshold such that agent $x$ has a minimum cost of $0$ in $C_2$. However, the cost of agent $x$ in coalition $C_1$ is strictly positive as two of the five agents in $C_1$ have a higher distance than $\lambda =2$ to agent $x$. As a result, agent $x$ prefers to jump to coalition $C_3$ with $\lambda=2$.

    Further, with the \textsc{Average} cost function, the agents with values $8$ and $9$ want to jump to each other's coalition, but the agent with value $8$ would not agree to a swap with the agent with value $9$. In their current coalitions, agent $a_8$ has a cost of $\frac{4+4+1+1}{|C_1|-1}=2.5$ and agent $a_9$ has a cost of $\frac{4+4}{|C_3|-1}=4$. However, both agents have lower cost when jumping to each other's coalitions, since agent $a_8$ then has a cost of $\frac{3+3+1}{|C_3|}\approx 2.67$ and agent $a_9$ a cost of $\frac{5+5+2+2+1}{|C_1|}=3$. However, agent $a_8$ would not agree to swap with agent $a_9$ since its cost in coalition $C_3$ would be $\frac{3+3}{|C_3\setminus \{a_9\}|}=3$ which is higher than the cost of $\approx 2.67$ agent $a_8$ has in coalition $C_1$. That shows that stability concepts of \textsc{Jump} and \textsc{Swap} are not necessarily related with \textsc{Average} cost function. Moreover, similar examples can be found for the other cost functions.
\end{example}

\section{Omitted Details from Section 2}\label{appx:swapStability}
\lemmaSwapMFHG*
\begin{proof}[Proof of \Cref{lemma:swap-mfhg}]
Consider any MFHG with agents $[n]$ and valuation functions $v_i\colon [n] \to \mathbb{R}_{\geq 0}$ for all $i \in [n]$ with $v_i(j) = v_j(i)$ for all $i,j \in [n]$. We show that the social costs are a potential function for this game. For that, let $\mathcal{C}$ be a coalition structure in which two agents $a,b \in [n]$ have an improving swap. Let $\mathcal{C}^\prime$ be a coalition structure after agents $a$ and $b$ swapped. We show that $\SW(\mathcal{C}) < \SW(\mathcal{C}^\prime)$.

For the proof, let $S_a \coloneqq |\mathcal{C}(a)| = |\mathcal{C}^\prime(b)|$ and $S_b \coloneqq |\mathcal{C}(b)| = |\mathcal{C}^\prime(a)|$ be the sizes of the two coalitions of $a$ and $b$ before the swap (including $a$ and $b$). Note that both coalitions are not singletons, since no agent would swap into a singleton coalition (to get zero utility) in a setting, where all valuations are non-negative. With that, we can state the difference in the social welfare of the coalition structures $\mathcal{C}$ and $\mathcal{C}^\prime$, which is
\begin{align}\label{eq1:swap-mfhg}
    \SW(\mathcal{C}^\prime) = \SW(\mathcal{C}) &- \left(2\frac{v_a(\mathcal{C}(a)\setminus\{a\})}{S_a-1} + 2\frac{v_b(\mathcal{C}(b)\setminus\{b\})}{S_b-1}\right)\\\nonumber
    &+ \left(2\frac{v_a(\mathcal{C}^\prime(a)\setminus\{a\})}{S_b-1} + 2\frac{v_b(\mathcal{C}^\prime(b)\setminus\{b\})}{S_a-1}\right),
\end{align}
where we use $v_x(N) \coloneqq \sum_{i \in N}v_x(i)$ for all agents $x$ and set of agents $N$. Now, we use the fact that both agents $a$ and $b$ strictly improve their utility by the swap, i.e., 
\begin{align*}
    \frac{v_a(\mathcal{C}(a)\setminus\{a\})}{S_a-1} <& \frac{v_a(\mathcal{C}^\prime(a)\setminus\{a\})}{S_b-1}\\
    \frac{v_b(\mathcal{C}(b)\setminus\{b\})}{S_b-1} <& \frac{v_b(\mathcal{C}^\prime(b)\setminus\{b\})}{S_a-1}.
\end{align*}
Together with \Cref{eq1:swap-mfhg}, we get that $\SW(\mathcal{C}^\prime) -\SW(\mathcal{C}) > 0$.
\end{proof}

\corSwapPotentialGames*
\begin{proof}[Proof of \Cref{cor:swap-potential-games}]
The equilibrium existence for \textsc{Max-Swap} games follows from reducing to the model by~\citet{BiloSchellingGamesContinuous2023}. \\
For \textsc{Avg-} and \textsc{Cutoff$_\lambda$-Swap} games we use \Cref{lemma:swap-mfhg} together with a transformation from the cost variants that we use to the respective utility variants of MFHG. In both cases, this transformation simply subtracts the cost of an agent from the overall maximum cost to get the utility. Note that \Cref{lemma:swap-mfhg} actually only covers the \textsc{UIS} case. However, the proof is the same for the \textsc{HIS} setting, since in both cases, no agent in a singleton coalition will ever be part of a swap.
All three cases employ potential functions to show equilibrium existence, thereby showing the FIP.
\end{proof}

\theoremSwapSortedEqExists*
\begin{proof}[Proof of \Cref{theorem:swap-sorted-eq-exists}]\let\qed\relax
In general, let $(n, (v_i)_{i \in [n]}, k, (k_i)_{i \in [k]})$ be a \textsc{Swap} game instance and $\mathcal{C}\coloneqq \{C_i\}_{i \in [k]}$ some sorted coalition structure. We assume for contradiction that there are two agents $x,y \in [n]$ with (wlog) $v_x < v_y$ that want to swap. In the following, we show independently for all three cost functions that this leads to a contradiction. For that, let $\mathcal{C}^\prime$ be the coalition structure $\mathcal{C}$ after the swap of agents $x$ and $y$, and $X \coloneqq \mathcal{C}(x) \setminus \{x\}$ and $Y \coloneqq \mathcal{C}(y) \setminus \{y\}$ the sets of other agents in the coalitions of agent $x$ and agent $y$ before the swap, respectively.
\end{proof}

\begin{proof}[Proof for \textsc{Max}:]
Let $a$ and $b$ be the minimum and maximum value in the coalition $\mathcal{C}(x)$, and $c$ and $d$ be the minimum and maximum value in the coalition $\mathcal{C}(y)$. Since $\mathcal{C}$ is sorted, we know that
\begin{align}\label{eq1-max:theorem:swap-sorted-eq-exists}
    a \leq v_x \leq b \leq c \leq v_y \leq d.
\end{align}
Further, we can use the fact that agent $x$ strictly decreases its cost by this swap, i.e., 
\begin{align}\label{eq2-max:theorem:swap-sorted-eq-exists}
    \max\{\dist(x,a),\dist(x,b)\} > \max\{\dist(x,c),\dist(x,d)\}.
\end{align}
Using \Cref{eq1-max:theorem:swap-sorted-eq-exists}, we know that $\max\{\dist(x,c),\dist(x,d)\} = \dist(x,d)$. Further, since $v_x \leq b \leq d$, \Cref{eq2-max:theorem:swap-sorted-eq-exists} implies
    $$\dist(x,a) > \dist(x,d).$$
Following the same arguments, $\dist(y,d) > \dist(y,a)$ holds too. With that and, again, \Cref{eq1-max:theorem:swap-sorted-eq-exists}, we can defer the following contradiction
    $$\dist(x,a) > \dist(x,d) \geq \dist(y,d) > \dist(y,a) \geq \dist(x,a).$$
\end{proof}

\begin{proof}[Proof for \textsc{Average}:]
We show that if one agent, wlog $x$, has an incentive to swap with agent $y$, then agents $y$'s cost does not decrease with this swap. For this technical proof, we need a value $z$ that lies between the coalitions of agents $x$ and $y$, i.e., it holds that $\max_{i \in \mathcal{C}(x)}v_i \leq z \leq \min_{i \in \mathcal{C}(y)}v_i$. Further, we partition the coalition of agents $x$ into a sets $X_L$ and $X_R$ of agents with a smaller / higher value than $v_x$. We do the same for the coalition of $y$ to get the sets $Y_L$ and $Y_R$ respectively. By the assumption that agent $x$ wants to swap with $y$, we know that 
\begin{align}\label{eq1-avg:theorem:swap-sorted-eq-exists}
    \frac{\dist(x,X_L)}{|X|} + \frac{\dist(x,X_R)}{|X|} > \dist(x,z) + \frac{\dist(z,Y)}{|Y|},
\end{align}
where $\dist(a,S) \coloneqq \sum_{i \in S} \dist(a,S)$ is the sum of distances from an agent $a$ to all agents of a set $S$. Since we know that $\dist(x,X_R) \leq |X_R|\dist(x,z)$, we know that
\begin{align}\label{eq2-avg:theorem:swap-sorted-eq-exists}
    \dist(x,X_L) > |X_L|\dist(x,z) + |X|\frac{\dist(z,Y)}{|Y|},
\end{align}
since otherwise it would contradict \Cref{eq1-max:theorem:swap-sorted-eq-exists}. The intuition of \Cref{eq2-avg:theorem:swap-sorted-eq-exists} is that the main reason for agent $x$ having the incentive to swap with some agent with a higher value, is that the agents in agent $x$'s coalition with smaller values than agent $x$ increase agent $x$'s cost a lot. We use this statement and show these agents also increase the cost of agent $y$ after the swap that much that agent $y$ does not want to swap with agent $x$. Formally, using \Cref{eq2-avg:theorem:swap-sorted-eq-exists}, yields the following lower bound on the cost of agent $y$ after the swap:
\begin{align*}
    \cost(y,X) =\;& \dist(y,z) + \frac{\dist(z,X_R)}{|X|} + \frac{|X_L|\dist(z,x)}{|X|} + \frac{\dist(x,X_L)}{|X|}\\
    \geq\;& \dist(y,z) + \frac{\dist(x,X_L)}{|X|}\\
    >\;& \dist(y,z) + \frac{|X_L|\dist(x,z)}{|X|} + \frac{\dist(z,Y)}{|Y|}\\
    \geq\;& \dist(y,z) + \frac{\dist(z,Y)}{|Y|},
\end{align*}
which is an upper bound on the cost of agent $y$ before the swap
\begin{align*}
    \cost(y,Y) = \frac{1}{|Y|}\left(|Y_L|\dist(y,z) - \dist(z,Y_L) + \dist(z,Y_R) - |Y_R|\dist(y,z)\right).
\end{align*}
This shows that, if some agent $x$ wants to swap with some agent $y$ in a sorted coalition structure, then the cost of agent $y$ increases by this swap. Therefore, there is no improving swap in a sorted coalition structure in an \textsc{Avg-Swap} game.
\end{proof}

\begin{proof}[Proof for \textsc{Cutoff$_\lambda$}:]
First note that both $x$ and $y$ need to have at least one friend in their new coalition after the swap, as they would not swap to coalition with only enemies. This implies that $N_x^-(X)$ only includes agents in $X$ with a smaller value as $x$, and that $N_y^-(Y)$ only includes agents in $Y$ with a higher value as $y$. Therefore, it holds that $N_x^-(X) \subseteq N_y^-(X)$ and $N_y^-(Y) \subseteq N_x^-(Y)$. We use these inclusions to show the following statement: if one agent, wlog $x$, has an incentive to swap with agent $y$, then agent $y$'s cost does not decrease with this swap. For that, consider the following chain of inequalities
\begin{align*}
    \cost(y,\mathcal{C}^\prime(y)) =\;& \frac{|N_y^-(X)|}{|X|} \overset{(1)}{\geq} \frac{|N_x^-(X)|}{|X|}\\ \overset{(2)}{>}\;& \frac{|N_x^-(Y)|}{|Y|} \overset{(3)}{\geq} \frac{|N_y^-(Y)|}{|Y|} = \cost(y,\mathcal{C}(y))
\end{align*}
where Inequality (1) and (3) hold because $N_x^-(X) \subseteq N_y^-(X)$ and $N_y^-(Y) \subseteq N_x^-(Y)$ respectively, and Inequality (2) holds due to the incentive of agent $x$ to swap. This shows that the cost of agent $y$ after the swap is higher than the cost before the swap. Therefore, there is no improving swap in a sorted coalition structure in a \textsc{Cutoff$_{\lambda}$-Swap} game.
\end{proof}

\lemmaSwapUnsortedOpt*
\begin{proof}[Proof of \Cref{lemma:swap-unsorted-opt} continued]
    We show that the unsorted coalition structure
    \begin{align*}
        \mathcal{C}^* \coloneqq  \bigl \{\{1,1,3,3\},\{2,2,2,2,2\} \bigr \}
    \end{align*}
    is the social optimum and a swap equilibrium.

    First note that the coalition structure $\mathcal{C}^*$ (defined above) is a swap equilibrium as all but one coalition include only agents with minimum cost of $0$. To show that this is indeed also the optimum, we consider the three different cost functions separately and calculate the cost for all possible coalition structures (up to the symmetry of $1$ and $3$).

    The \textsc{Average} costs are:
    \begin{itemize}
        \item $\mathcal{C}^* \coloneqq \big\{\{1,1,3,3\}\{2,2,2,2,2\}\big\}$ has social cost $5+\frac{1}{3}$
        \item $\big\{\{1,1,2,3\},\{2,2,2,2,3\}\big\}$ has social cost $6+\frac{2}{3}$
        \item $\big\{\{1,1,2,2\},\{2,2,2,3,3\}\big\}$ has social cost $5+\frac{2}{3}$
        \item $\big\{\{1,2,2,3\},\{1,2,2,2,3\}\big\}$ has social cost $8.0$
        \item $\big\{\{1,2,2,2\},\{1,2,2,3,3\}\big\}$ has social cost $7.0$
        \item $\big\{\{2,2,2,2\},\{1,1,2,3,3\}\big\}$ has social cost $6.0$
    \end{itemize}

    The \textsc{Max} costs are:
    \begin{itemize}
        \item $\mathcal{C}^* \coloneqq \big\{\{1,1,3,3\},\{2,2,2,2,2\}\big\}$ has social cost $8$
        \item $\big\{\{1,1,2,3\},\{2,2,2,2,3\}\big\}$ has social cost $12$
        \item $\big\{\{1,1,2,2\},\{2,2,2,3,3\}\big\}$ has social cost $9$
        \item $\big\{\{1,2,2,3\},\{1,2,2,2,3\}\big\}$ has social cost $13$
        \item $\big\{\{1,2,2,2\},\{1,2,2,3,3\}\big\}$ has social cost $12$
        \item $\big\{\{2,2,2,2\},\{1,1,2,3,3\}\big\}$ has social cost $9$
    \end{itemize}

    The \textsc{Cutoff$_\lambda$} costs are:
    \begin{itemize}
        \item $\mathcal{C}^* \coloneqq \big\{\{1,1,3,3\},\{2,2,2,2,2\}\big\}$ has social cost $2+\frac{2}{3}$
        \item $\big\{\{1,1,2,3\},\{2,2,2,2,3\}\big\}$ has social cost $5+\frac{1}{3}$
        \item $\big\{\{1,1,2,2\},\{2,2,2,3,3\}\big\}$ has social cost $5+\frac{2}{3}$
        \item $\big\{\{1,2,2,3\},\{1,2,2,2,3\}\big\}$ has social cost $6+\frac{5}{6}$
        \item $\big\{\{1,2,2,2\},\{1,2,2,3,3\}\big\}$ has social cost $6.0$
        \item $\big\{\{2,2,2,2\},\{1,1,2,3,3\}\big\}$ has social cost $4.0$
    \end{itemize}

For all cost functions, the coalitions structure $\mathcal{C}^*$ has the smallest social cost and is the unique optimum.
\end{proof}

\section{Omitted Details from Section 3}
\label{appx:JumpStability}

\lemmaMaxJumpFIP*

\begin{proof}[Proof of \Cref{lemma:max-jump-fip}]
    We show that both kind of games are potential games. Both potential functions are based on the idea that the non-increasingly sorted cost vector of all agents decreases lexicographically with every improving jump. This always holds for the \textsc{HIS} setting. For \textsc{UIS} setting, however, we have to take the number of singleton coalitions into account.

    For both proofs, let $\mathcal{C}$ be a coalition structure, and $C,D \in \mathcal{C}$ two coalitions such that there is an agent $x$ with an improving jump from coalition $C$ to coalition $D$. Let $\mathcal{C}^\prime$ be the coalition structure $\mathcal{C}$ after agent $x$ jumped.

    For the \textsc{HIS} setting, let $\Phi_{\cost}(\mathcal{C})$ be the vector of the costs of all agents that is sorted non-increasingly. We prove that $\Phi_{\cost}(\mathcal{C}) > \Phi_{\cost}(\mathcal{C}^\prime)$, where the vectors are compared lexicographically. To do so, we show that all agents whose cost change have a cost of less than $\cost(x,C)$ after agent $x$'s jump. When agent $x$ jumps out of coalition $C$, the cost of all remaining agents in coalition $C$ cannot increase. Thus, these agents will not increase $\Phi_{\cost}$. However, the costs of some agents in coalition $D$ increase with the jump if $v_x \notin [\min_{i \in D}v_i;\max_{i \in D}v_i]$ (at least the cost of the agents with value $\max_{i \in D}v_i$ if $v_x < \min_{i \in D}v_i$ or the cost of the agents with values $\min_{i \in D}v_i$ if $v_x > \max_{i \in D}v_i$). But then, their costs are at most the cost of agent $x$ in coalition $D$, since the value of agent $x$ is one of the new extreme values in coalition $D$. As a result, $\Phi_{\cost}$ decreases with every improving jump and is thus a potential function.

    Note that, in the previous proof, we heavily used that we are in the \textsc{HIS} setting, as agent $x$ sets the cost of another agent to $0$ if $|C|=2$. In the \textsc{UIS} setting, this does not work anymore, because the cost of the other agent in $C$ would increase over $\cost(x,C)$, as it then is isolated in $C$. To fix this, we define the vector potential function $\Phi(\mathcal{C}) \coloneqq (|\mathcal{C}|_{\neq \emptyset},\Phi_{\cost}(\mathcal{C}))$ where $|\mathcal{C}|_{\neq \emptyset}$ is the number of non-empty coalitions in the given coalition structure, and $\Phi_{\cost}(\mathcal{C})$ is the non-increasingly sorted cost vector used in the \textsc{HIS} setting, meaning that agents in isolation have entry of $0$ in the cost vector (although their cost are not zero in the \textsc{UIS} setting). We show that every improving jump strictly decreases $\Phi$. If the jumping agent $x$ is not in a singleton coalition, then the number of coalitions stays the same and the $\Phi_{\cost}(\mathcal{C})$ decreases as it does in the \textsc{HIS} setting above. If, however, agent $x$ jumps out of a singleton coalition, then the number of non-empty coalitions decreases. Since no agent will ever jump into an empty coalition in the \textsc{UIS} setting, $\Phi_{\cost}(\mathcal{C})$ will never increase with an improving jump. As a result, $\Phi$ is a potential function.
\end{proof}

In order to prove \Cref{theorem:his-sorted-pne}, we first give some further properties of monotone cost functions in \Cref{obs:cutoff-sorted}~and~\Cref{lemma:left-right-improving-jump} that are directly implied by \Cref{def:monotone-cost-f}.

\begin{observation}
    \label{obs:cutoff-sorted}
    Let $\mathcal{I} \coloneqq (n,(v_i)_{i\in [n]},k)$ for some $k \geq 2$ be a game instance with monotone cost function $\cost_m$. Further, let $\mathcal{C}\coloneqq \{C_i\}_{i \in [k]}$ be a sorted coalition structure of $\mathcal{I}$ and $x$ be an agent in some coalition $C_i$. Then
    \begin{enumerate}
        \item[(i)] $\cost_m(x,C_j) \geq \cost_m(R(C_i),C_j)$ for all $j \in [k]$ with $j > i$,\\
        $\cost_m(x,C_j) \geq \cost_m(L(C_i),C_j)$ for all $j \in [k]$ with $j < i$,
        \item[(ii)] $\cost_m(x,C_{i+1}) \leq \cost_m(x,C_{i+2}) \leq \dots \leq \cost_m(x,C_k)$, and\\
        $\cost_m(x,C_{i-1}) \leq \cost_m(x,C_{i-2}) \leq \dots \leq \cost_m(x,C_1)$.
    \end{enumerate}
\end{observation}

\lemmaImprovingMoveStable*
\begin{proof}
    If $\mathcal{C}$ is not a jump equilibrium, then there is an agent~$x$ who has the incentive to jump from its current coalition $C_i$ to some other coalition $C_j$, i.e., $\cost_m(p,C_i) > \cost_m(p,C_j)$. Since the cost function is monotone, property (iii) of \Cref{def:monotone-cost-f} directly implies $\cost_m(R(C_i),C_i) > \cost_m(R(C_i),C_j)$ (if $i < j$) or $\cost_m(L(C_i),C_i) > \cost_m(L(C_i),C_j)$ (if $j < i$).
\end{proof}

Now, we are ready to prove that the properties of monotone cost functions are sufficient for the algorithm of Milchtaich and Winter to work properly.

\theoremHisSortedPne*
\begin{proof}
    In order to construct such a PNE for $\mathcal{I}$, we start with a \enquote{right-heavy} coalition structure $\mathcal{C}_{\text{start}}$ (defined below) and iteratively execute left-improving moves until no left-improving move is possible. The start coalition structure $\mathcal{C}_{\text{start}} \coloneqq \{C_i\}_{i \in [k]}$ is the coalition structure where the coalition $C_i$ (for $i \in [k-1]$) only contains the agents with the $i$-th lowest value and all remaining agents are assigned to coalition $C_k$, i.e.,
        $$\mathcal{C}_{\text{start}} \coloneqq \big\{\{v_1\},\dots,\{v_{k-1}\},\{v_k,\dots,v_n\}\big\}.$$
    Clearly, $\mathcal{C}_{\text{start}}$ is sorted and it stays sorted under left-improving moves. It is therefore valid to also number coalitions during the algorithm from $C_1$ to $C_k$ according to the ordering of their values. 
    
    To prove that this algorithm works under the monotone cost function $\cost_m$, we show by induction that there is no coalition structure during the algorithm admitting a right-improving move. By \Cref{lemma:left-right-improving-jump}, this proves that the final coalition structure of this algorithm, for which by definition also no left-improving move is possible, is indeed a jump equilibrium. Note that this algorithm terminates at latest when no left move is possible without hurting the \textsc{HIS} setting, that is the coalition in which the $k-1$ highest values form $k-1$ singleton coalitions and the remaining ones moved all the way to the left coalition $C_1$.

    In the base case, i.e., in $C_{\text{start}}$, no right move is possible because all agents in the coalitions $C_1,\dots,C_{k-1}$ have minimum cost (due to the \textsc{HIS} setting) and all other agents are in the right-most coalition. Now consider a sorted coalition structure $\mathcal{C}\coloneqq \{C_i\}_{i \in [k]}$ which does not admit any right-improving move, but at least one left-improving move. Let $L(C_i)$ be the agent with a left-improving move in $\mathcal{C}$, and $\mathcal{C}^\prime \coloneqq \{C^\prime_i\}_{i \in [k]}$ the coalition structure $\mathcal{C}$ after the left-improving move of agent $L(C_i)$. See \Cref{fig:cutoff-pne-proof} for an illustration.

    The only coalitions that changed are $C_i$ and $C_{i-1}$. Therefore and with \Cref{obs:cutoff-sorted}, we only need to show that none of the agents $R(C^\prime_{i-2})$, $R(C^\prime_{i-1})$ and $R(C^\prime_i)$ have a right-improving move. 
    If there was some other right-improving move, it would have already been in $\mathcal{C}$ which is a contradiction to the induction hypothesis. We can directly rule out that $R(C^\prime_{i-1})$ has a right-improving move, since this would undo the left-improving move from before (from $\mathcal{C}$ to $\mathcal{C}^\prime$). We will prove that there are also no right-improving moves for the remaining two candidates in the following two paragraphs.

    If $i = k$, $R(C^\prime_i)$ cannot move to the right and thus can not have a right-improving move. Otherwise, we use that $[r \coloneqq ]R(C^\prime_i)=R(C_i)$ and show that the cost of agent $r$ did not increase in that step, i.e., $\cost_m(R(C^\prime_i),C^\prime_i) \leq \cost_m(R(C_i),C_i)$. With that, the following chain of inequalities would hold
    \begin{align*}
        \cost_m(R(C^\prime_i),C^\prime_i) \leq\;& \cost_m(R(C_i),C_i) \\
        \overset{(1)}{\leq}\;& \cost_m(R(C_i),C_{i+1}) \overset{(2)}{=} \cost_m(R(C^\prime_i),C^\prime_{i+1}),
    \end{align*}
    where (1) holds because there was no right-improving for $r$ in $\mathcal{C}$ by the induction hypothesis, and (2) holds because the coalition to $r$'s right (which exists since $i < k$) did not change either, i.e., $C_{i+1} = C^\prime_{i+1}$. This implies that there is no right-improving move for $R(C^\prime_i)$ in $\mathcal{C}^\prime$. 
    To show the yet missing inequality $\cost_m(R(C^\prime_i),C^\prime_i) \leq \cost_m(R(C_i),C_i)$, we use the fact that $\cost_m$ is monotone and instantiate part (ii) of \Cref{def:monotone-cost-f} with $x \coloneqq r$, $C \coloneqq C^\prime_i \setminus \{r\}$ and $D \coloneqq \{R(C^\prime_{i-1})\} = \{L(C_i)\}$. With that, it holds that
        $$\cost_m(r,C^\prime_i \setminus \{r\}) \leq \cost_m(r,C^\prime_i \setminus \{r\}\cup \{L(C_i)\}).$$
    We observe that $\cost_m(r,C^\prime_i \setminus \{r\}) = \cost_m(r,C^\prime_i)$ and $\cost_m(r,C^\prime_i \setminus \{r\}\cup \{L(C_i)\}) = \cost_m(r,C_i \setminus \{r\}) = \cost_m(r,C_i)$. Therefore, we get $\cost_m(R(C^\prime_i),C^\prime_i) \leq \cost_m(R(C_i),C_i)$. As a result, there is no right-improving move for $R(C^\prime_i)$ in $\mathcal{C}^\prime$.

    In case $R(C^\prime_{i-2})$ exists, i.e., $i \geq 3$, we show that the cost for agent $r \coloneqq R(C_{i-2}) = R(C^\prime_{i-2})$ in the coalition to its right do not decrease with the left-improving move of $L(C_i)$, i.e., we show $\cost_m(r,C^\prime_{i-1}) \geq \cost_m(r,C_{i-1})$. With that, the following chain of inequalities holds
    \begin{align}\label{eq:theorem:his-sorted-pne}
        \cost_m(r,C^\prime_{i-2}) \overset{(1)}{=}\;& \cost_m(r,C_{i-2})\\ \nonumber\overset{(2)}{\leq}\;& \cost_m(r,C_{i-1}) \leq \cost_m(r,C^\prime_{i-1})
    \end{align}
    where (1) holds because the coalition $C_{i-2}$ did not change in this step, i.e., $C_{i-2} = C^\prime_{i-2}$, and (2) holds because there was no right-improving move for $r$ in $\mathcal{C}$ by the induction hypothesis. This implies that there is no right-improving move for $r$ in $\mathcal{C}^\prime$. To show the remaining inequality $\cost_m(r,C^\prime_{i-1}) \geq \cost_m(r,C_{i-1})$, we use the fact that $\cost_m$ is monotone and instantiate part (ii) of \Cref{def:monotone-cost-f} with $x \coloneqq r$, $C \coloneqq C_{i-1} = C^\prime_{i-1} \setminus \{R(C^\prime_{i-1})\}$ and $D \coloneqq \{R(C^\prime_{i-1})\} = \{L(C_i)\}$. With that, it holds that
        $$\cost_m(r,C_{i-1}) \leq \cost_m(r,(C^\prime_{i-1} \setminus \{R(C^\prime_{i-1})\})\cup \{R(C^\prime_{i-1})\}).$$
    Since $\cost_m(r,(C^\prime_{i-1} \setminus \{R(C^\prime_{i-1})\})\cup \{R(C^\prime_{i-1})\}) = \cost_m(r,C^\prime_{i-1})$, this directly implies $\cost_m(r,C^\prime_{i-1}) \geq \cost_m(r,C_{i-1})$. With \Cref{eq:theorem:his-sorted-pne}, there is no right-improving move for $r=R(C^\prime_{i-1})$ in $\mathcal{C}^\prime$.

    To conclude, the algorithm starts with a coalition structure that is sorted and has no right-improving moves, and keeps these invariants under left-improving moves given a monotone cost function. When no left-improving moves are available, the invariant indicates that also no right-improving moves are possible. By \Cref{lemma:left-right-improving-jump}, this proves that the final coalition structure of our algorithm is indeed a jump equilibrium of the given game instance.
\end{proof}

\lemmaMonotoneCostF*
\begin{proof}[Proof for \Cref{lemma:monotone-cost-f}]\let\qed\relax
    We split the proof of this lemma into three different parts; each for one of the three cost functions. For that, let $x$ and $y$ be two agents and $C$ and $D$ be two coalitions. Every part (i, ii, iii) of \Cref{def:monotone-cost-f} considers two different cases. In the following proofs, we only consider the case where
        $$v_x \leq v_y \leq v_{L(C)} \leq v_{R(C)} \leq v_{L(D)}$$
    since the other one can be proven with the analogous arguments.
\end{proof}
\begin{proof}[Proof for Cutoff:]
    For part (i), we show that
    \begin{align}\label{eq1:cutoff-monotone}
        \cost(x,C) = \frac{|N^-_x(C)|}{|C|} \geq \frac{|N^-_y(C)|}{|C|} = \cost(y,C)
    \end{align}
    
    by proving $N^-_y(C) \subseteq N_x^-(X)$, i.e., every enemy of $y$ in $C$ is also an enemy of $x$ in $C$. Let $c \in N^-_y(C)$ be an enemy of $y$ in $C$. Then, it holds that
        $$v_c - v_x \geq v_c - v_y > \lambda,$$
    because we assumed $v_x \leq v_y$ and we know that $c \in N^-_y(C)$. Therefore, agent $c$ is an enemy of $x$, which implies $N^-_y(C) \subseteq N^-_x(C)$ and thus \Cref{eq1:cutoff-monotone}.

    For part (ii), we have to prove that $\cost(x,C) \leq \cost(x,C\cup D)\leq \cost(x,D)$. For that, we make a case distinction on where the value $v_x + \lambda$ lies. If $v_x + \lambda < v_{L(C)}$, then all agents in $C \cup D$ are enemies of agent $x$. Thus, all costs are $1$ and thus satisfy the chain of inequalities. If $v_x + \lambda \in [v_{L(C)},v_{R(C)}]$, i.e., lies within the range of coalition $C$, then all agents in coalition $D$ are enemies of agent $x$. Therefore, $\cost(x,D)=1$ holds, which is a trivial upper bound on the other two costs. To see that
        $$\cost(x,C) = \frac{|N_x^-(C)|}{|C|} \leq \frac{|N_x^-(C)|+|D|}{|C|+|D|} = \cost(x,C\cup D),$$ 
    is correct, note that $\frac{a}{b} \leq \frac{a+x}{b+x}$ holds for all $0 \leq a \leq b$ with $b \neq 0$ and all $x \geq 1$. In the last case, let $v_x + \lambda > v_{R(C)}$. Then, all agents in coalition $C$ are friends of agent $x$. Therefore, $\cost(x,C)=0$ holds, which is a trivial lower bound on the other two costs. The remaining inequality to prove part (ii) of \Cref{def:monotone-cost-f} holds because $|C| \geq 1$ and
        $$\cost(x,C\cup D) = \frac{|N_x^-(D)|}{|D|+|C|} \leq \frac{|N_x^-(D)|}{|D|} = \cost(x,D).$$ 

    For part (iii), let $c \in C$ be an agent with $\cost(c,C) > \cost(c,D)$. From part (i), we know that $\cost(c,D) \geq \cost(R(C), D)$. In order to prove that $\cost(R(C), C) > \cost(R(C), D)$, we only need to show that $\cost(R(C), C) \geq \cost(c, C)$ since
    \begin{align}\label{eq3:cutoff-monotone}
        \cost(R(C),C) \geq \cost(c,C) > \cost(c,D) \geq \cost(R(C),D)
    \end{align}
    holds afterward.
    For that, we first show that $c$ and $R(C)$ have to be friends. If $c$ and $R(C)$ were enemies, we knew that also all agents $x \in D$ were enemies of agent $c$ since we assumed the coalition $D$ left of coalition $C$. But then, agent $c$ had a cost of $\cost(c,D)=1$ in $D$ and thus no incentive to move to $D$.

    As a result, agents $c$ and $R(C)$ have to be friends. Therefore, there is no enemy of agent $c$ with a value in $[v_c;v_{R(C)}]$. With that, we show that all enemies of agent $c$ in coalition $C$ are also enemies of agent $R(C)$. For that, let $x \in N^-_c(C)$ be an enemy of agent $c$ in coalition $C$. As $v_x \notin [v_c;v_{R(C)}]$, it holds that $v_x \leq v_c \leq v_{R(C)}$. With that, we get that agent $x$ is also an enemy of agent $R(C)$ and thus $N_c^-(C) \subseteq N_{R(C)}^-(C)$. As a result, it holds that
    \begin{align*}
        \cost(R(C),C)
        = \frac{|N^-_{R(C)}(C)|}{|C|}
        \geq \frac{|N^-_{c}(C)|}{|C|}
        = \cost(c,C)
    \end{align*}
    which completes the proof of \Cref{eq3:cutoff-monotone}. This shows that $R(C)$ also wants to move to $D$.
\end{proof}
\begin{proof}[Proof for Average:]
    As a first technical observation, we can rewrite the cost of an agent $a \in \{x,y\}$ in a coalition $S \in \{C,D\}$ with $a \notin S$ as follows:
    \begin{align}\label{eq:avg-lemma:monotone-cost-f}
        \cost(p,S) =\;& \frac{1}{|S|}\sum_{e \in S}|v_e-v_a| \\\ \nonumber
        =\;& \frac{1}{|S|}\sum_{e \in S}(v_e-v_a) \\ \nonumber
        =\;& \avg(S) - v_a,
    \end{align}
    where $\avg(\cdot)$ is the average value over a given set of agents.

    To show part (i), remember that we assumed $v_x \leq v_y$, which directly implies:
        $$\cost(y,C) = \avg(C) - v_y \overset{v_y \geq v_x}{\leq} \avg(C) - v_x = \cost(x,C).$$

    To prove part (ii), i.e.,
        $\cost(x,C) \leq \cost(x,C\cup D) \leq \cost(x,D)$
    for the average cost function, we show that (1) $\cost(x,C\cup D)$ is a convex combination of $\cost(x,C)$ and $\cost(x,D)$, and (2) that $\cost(x,C) \leq \cost(x,D)$. For (1), consider
    \begin{align*}
        \cost(x,C\cup D) &\;= \frac{1}{|C|+|D|}\sum_{e \in C \cup D}|v_e-v_x|\\
        &\;= \frac{1}{|C|+|D|}\sum_{e \in E}|v_e-v_x| + \frac{1}{|C|+|D|}\sum_{e \in E}|v_e-v_x| \\
        &\;= \frac{|C|}{|C|+|D|}|\avg(C)-v_x| + \frac{|D|}{|C|+|D|}|\avg(D)-v_x| \\
        &\;= \frac{|C|}{|C|+|D|}\cost(x,C) + \frac{|D|}{|C|+|D|}\cost(x,D).
    \end{align*}
    For (2), note that the average of a given (multi)set always lies between the minimum and the maximum value in that (multi)set. Since we assume $v_{R(C)} \leq v_{L(D)}$, we know $\avg(C) \leq \avg(D)$ and thus
        $$\cost(x,C) = \avg(C) - v_x \leq \avg(D) - v_x = \cost(x,D),$$
    which concludes the proof for part (ii). 

    The proof for part (iii) is already done by \citet[Lemma 2.5]{MilchtaichStabilitySegregationGroup2002}.
\end{proof}
\begin{proof}[Proof for Maximum:]
    Since all agents in the coalitions $C$ and $D$ have a greater (or equal) value than $x$ and $y$, we know for all $a \in \{x,y\}$ and $S \in \{C,D\}$ that
    \begin{align}\label{eq:max-lemma:monotone-cost-f}
        \cost(a,S) = \max_{s \in S} |v_s-v_a| = \max_{s \in S} (v_s-v_a) = \max S - v_a.
    \end{align}
    Then, the respective case of part (i) follows by $v_x \leq v_y$ as
        $$\cost(x,C) = v_{R(C)} - v_x \geq v_{R(C)} - v_y = \cost(y,C).$$
    For part (ii), we use the assumption that $v_{R(C)} \leq v_{R(D)}$ and get
        $$\cost(x,C) = \max C - v_x \leq \max D - v_x = \cost(x,D)$$
    as well as 
        $$\cost (x,C\cup D) = v_{R(C\cup D)} - v_x = v_{R(D)} - v_x = \cost(x,D).$$
    Therefore, it also holds that $\cost(x,C) \leq \cost(x,C\cup D)\leq \cost(x,D)$.

    For part (iii), let $c \in C$ be an agent with $\cost(c,C) > \cost(c,D)$. From part (i), we know that $\cost(c,D) \geq \cost(R(C), D)$. In order to prove that $\cost(R(C), C) > \cost(R(C), D)$, we only need to show that $\cost(R(C), C) \geq \cost(c, C)$. Using similar arguments as for \Cref{eq:max-lemma:monotone-cost-f}, we get that $\cost(R(C), C) = v_{R(C)} - v_{L(C)}$. This value has a lower bound of $v_c - v_{L(C)}$, since $v_{L(C)} \leq v_c$. In the remainder of this proof, we show that $v_c - v_{L(C)} = \cost(c,C)$. For any element in $C$ the cost is either the distance to agent $L(C)$ or the distance to agent $R(C)$. Assume for contradiction that $\cost(c,C) = v_{R(C)} - v_c$. Since we assume $v_{R(C)} \leq v_{R(D)}$, we know that $\cost(c,C) \leq v_{R(D)} - v_c = \cost(c,D)$, which contradicts the premise that $\cost(c,C) > \cost(c,D)$. As a result, the cost of $c$ has to be the distance to agent $L(C)$, i.e., $\cost(c,C) = v_c - v_{L(C)}$. To conclude, we get the following chain of inequalities:
    \begin{align*}
        \cost(R(C), D) \overset{(1)}{\leq}\;& \cost(c,D) \overset{(2)}{<} \cost(c,C) \\
        \overset{(3)}{=}\;& v_c - v_{L(C)} \overset{v_c \leq v_{R(C)}}{\leq} v_{R(C)} - v_{L(C)} \\
        =\;& \cost(R(C), C),
    \end{align*}
    where (1) holds because of part (i) and $v_c \leq v_{R(C)}$, (2) by definition of agent $c$ and (3) by the previous observation. This shows that agent $R(C)$ has the incentive to jump to coalition $D$ if this holds for some agent $c \in C$ too. 
\end{proof}

\lemmaAvgUnsortedJumpEq*
\begin{proof}[Proof of \Cref{lemma:avg-unsorted-jump-eq}]
    Consider the following symmetric instance,
    $$ \big\{\{1, 1, 3, 3, 3, 3\}, \{2, 2, 2, 2, 4, 4\}\big\},$$
    having two coalitions that overlap each other. 
    It is a jump equilibrium, as no agent has an improving move.
    Since it is a symmetric instance, we only need to check agents 1 and 3: agent 1 has cost $\frac{4\cdot2}{5}=1.6$, and would have cost $\frac{4\cdot1+2\cdot3}{6}\approx 1.67$ in the other coalition; and agent 3 has cost $\frac{2\cdot2}{5}=0.8$, and would have cost $\frac{6\cdot1}{6}=1$ in the other coalition.
\end{proof}
This overlapping structure seems to be generalizable to create more overlaps, e.g., to
$$\big\{\{\overbrace{1\ldots1}^3, \overbrace{3\ldots3}^4,\overbrace{5\ldots5}^5\}, \{\overbrace{2\ldots2}^5,\overbrace{4\ldots4}^4,\overbrace{6\ldots6}^3\}\big\}$$
with three parts per coalition, or 
$$\big\{\{\overbrace{1\ldots1}^4, \overbrace{3\ldots3}^6,\overbrace{5\ldots5}^6,\overbrace{7\ldots7}^7\}, \{\overbrace{2\ldots2}^7,\overbrace{4\ldots4}^6,\overbrace{6\ldots6}^6,\overbrace{8\ldots8}^4\}\big\}$$
with four parts per coalition. So far this structure seems to be the only jump stable with unsorted coalitions, but we were not able to show which properties an instance needs to have unsorted equilibria.

\lemmaCutoffjumpunsorted*
\begin{proof}[Proof of \Cref{lemma:cutoff-jump-unsorted} continued]
    In the following, we show that the coalition structure $\mathcal{C}^*$ is also optimum by calculating the social cost of all coalition structures (up to symmetry of the agents with values $1$ and $3$):
    \begin{itemize}
        \item $\{\{1\},\{1,2,2,2,3,3\}\}$ has social cost $\geq4.4$ (even for HIS)
        \item $\{\{2\},\{1,1,2,2,3,3\}\}$ has social cost $\geq 4.8$ (even for HIS)
        \item $\{\{1,1\},\{2,2,2,3,3\}\}$ has social cost $3.0$
        \item $\{\{1,2\},\{1,2,2,3,3\}\}$ has social cost $6.0$
        \item $\{\{1,3\},\{1,2,2,2,3\}\}$ has social cost $5.5$
        \item $\{\{2,2\},\{1,1,2,3,3\}\}$ has social cost $4.0$
        \item $\{\{1,1,2\},\{2,2,3,3\}\}$ has social cost $4+\frac{2}{3}$
        \item $\{\{1,1,3\},\{2,2,2,3\}\}$ has social cost $4.0$
        \item $\{\{1,2,2\},\{1,2,3,3\}\}$ has social cost $5+\frac{1}{3}$
        \item $\{\{1,2,3\},\{1,2,2,3\}\}$ has social cost $6+\frac{1}{3}$
        \item $\mathcal{C}^* = \{\{2,2,2\},\{1,1,3,3\}\}$ has social cost $2+\frac{2}{3}$
    \end{itemize}
    As $\mathcal{C}^*$ has the lowest cost of all coalition structures it is the optimum of this instance. That proves that there is an instance with an unsorted optimum.
\end{proof}

\lemmaMaxJumpUnsortedEq*
\begin{proof}[Full Proof of \Cref{lemma:max-jump-unsorted-eq}]
    We first analyze an instance for $k=2$ and $n=4k$ and then explain how this can be generalized for $k>2$ and $n>4k$. Consider the following instance:  
    $$\big\{\{L,L,H,H\},\{L,L,H,H\}\big\}.$$
    No agent has an improving jump, as its valuation in both coalitions is exactly the same (namely $H-L$).
    For larger $k>2$ and $n>4k$, we can generalize the instance in the following way:
    $$\big\{\overbrace{\{L,L,H,H\}}^{k-1},\{L,L,\overbrace{H\ldots H}^{n - 4k + 2}\}\big\}.$$
    To increase $k$, we can copy one of the coalitions, to increase $n$ we can duplicate one of the agents as many times as needed.
\end{proof}

\lemmaUnsortedUnstableMaxJumpOpt*
\begin{proof}[Proof of \Cref{lemma:unsorted-unstable-max-jump-opt}]
    This result is a modification of the instance in \Cref{lemma:swap-unsorted-opt}. By adding additional $2$'s to the right coalition, the improving jump of an agent of the left coalition results in an increase in social cost. For that, see that
    $$\big\{\{1,1,3,\mathbf{\vec{3}}\},\{\overbrace{2\ldots2}^7\}\big\}$$
    has social cost of $4\cdot2+7\cdot0=8$ while the coalition structure after the jump
    $$\big\{\{1,1,3\},\{\mathbf{3},\overbrace{2\ldots2}^7\}\big\}$$
    has a social cost of $3\cdot2+8\cdot1 = 14$. Even after a jump of another $3$ to the right
    $$\big\{\{1,1\},\{\mathbf{3,3},\overbrace{2\ldots2}^7\}\big\}$$
    the social cost stays higher at $2\cdot0+9\cdot1=9$. This is an equilibrium structure.
    The social optimum is unique up to symmetry (with cost 8), as we checked by brute force enumeration of all agent assignments. The intuition behind the optimality is as follows: it is the only structure where all seven $2$'s pay nothing, as they are on their own in one coalition. 
    If there are 2's in both coalitions, at least some of them pay extra, increasing social cost. 
    If then all 1's and 3's are in one coalition, this results in more cost than 8; 
    if 1's and 3's are split, all 2's and at least two other agents have to pay at least 1, resulting in a social cost of at least 9.
\end{proof}
\label{appx:avg-jump-opt}

\noindent\textbf{\textsc{Average-Jump} Optima Properties.\,}
Here, we analyze some properties of the optimal coalition structure in \textsc{Average-Jump} games for the case $k=2$.
Recall that agents are indexed in increasing order of their values.\\
For two consecutive agents $i$ and $i+1$, let $d_i=d(i,i+1)=v_{i+1}-v_i$ denote the difference between their values.
Then, for any pair of agents $i$ and $j$ with $j > i$, the distance $d(i,j)$ can be expressed as the sum of consecutive differences: $\sum_{p=i}^{j-1}d_p$.
Consequently, the total cost of a coalition structure $\mathcal{C} = \{C_1, C_2\}$, consisting of two coalitions $C_1$ and $C_2$ of sizes $m$ and $n - m$, respectively, can be written as a linear combination of consecutive value differences:
\begin{equation}
    \soc{\mathcal{C}} = \sum_{i=1}^{n-1}\alpha_i \cdot d_i,
\end{equation}
where $\alpha_i$ is the coefficient representing the contribution of $d_i$ to the total cost of $\mathcal{C}$.
Note that $d_i$ is defined only for $1\leq i \leq n-1$.

To characterize $\alpha_i$ more precisely, observe that $d_i$ contributes to the pairwise distance between any two agents $p$ and $q$ such that $p \leq i$ and $q > i$, provided both belong to the same coalition.
Let $\delta_i$ denote the number of agents $p \leq i$ (including agent $i$ itself) that belong to coalition $C_1$.
Then, there are $m - \delta_i$ agents in $C_1$ that come after $i$.
Similarly, there are $i - \delta_i$ agents (including $i$) before $i$ in $C_2$, and $n - m - i + \delta_i$ agents after $i$ in $C_2$.
Hence, $d_i$ contributes to $2\delta_i \cdot (m - \delta_i)$ pairwise distances within $C_1$ and to $2(i - \delta_i)\cdot (n - m - i + \delta_i)$ pairwise distances within $C_2$.
As a result, $\alpha_i$ depends on both $i$ and $\delta_i$ and can be more accurately expressed as $\alpha_i(\delta_i)$, given by the following formulation:
\begin{align}
\label{eq:alpha-def}
    \alpha_i(\delta_i) &= \frac{2\delta_i\cdot \bigl(m-\delta_i\bigr)}{m-1} + \frac{2\big(i-\delta_i\big)\cdot \big(n-m-i+\delta_i\big)}{n-m-1} \\ \notag
    & = 2\biggl(i +  \frac{\delta_i\cdot (1-\delta_i)}{m-1} + \frac{(i-\delta_i)\cdot (1-i+\delta_i)}{n-m-1} \biggr).
\end{align}
For simplicity, we write $\alpha_i$ instead of $\alpha_i(\delta_i)$ when $\delta_i$ is fixed. 
Observe that $\delta_i \le \min(i, m)$ for every $i$.
Moreover, wlog, we assume $m \le \nh$, since the case $m > n/2$ is symmetric.
Consequently, we always treat $C_1$ as the smaller coalition.\\

\begin{example}
First, we illustrate an example to clarify the notations introduced above.  
Consider the instance from~\Cref{lemma:avg-max-jump-pos}, consisting of 7 agents with values $(v_i)_{i\in [7]}\coloneqq (1, 1, 1, 4, 6, 8, 8)$, and the coalition structure
\[
\mathcal{C} = \bigl\{\{1,1,1,6\}, \{4,8,8\}\bigr\}, \quad C_1 = \{4,8,8\}, \; C_2 = \{1,1,1,6\}.
\]
Since $d_1 = d_2 = d_6 = 0$, these differences do not contribute to $\soc{\mathcal{C}}$, so we only consider $i=3,4,5$ in the calculations.  
By definition, we have $\delta_3=0$, $\delta_4=1$, and $\delta_5=1$.
Using~\Cref{eq:alpha-def}, we then compute:
\begin{align*}
    &\alpha_3(0)=2\left(3-\frac{3\cdot2}{3}\right)=2, \; \alpha_4(1)=2\left(4-\frac{3\cdot2}{3}\right)=4 \\
    &\alpha_5(1)=2\left(5-\frac{4\cdot3}{3}\right)=2.
\end{align*}
On the other hand,
\begin{align*}
    &\soc{C_1} = \frac{2\cdot2\cdot (8-4)}{2}= 2\bigl((8-6)+(6-4)\bigr) = 2(d_5+d_4),\\
    &\soc{C_2} = \frac{2\cdot3\cdot (6-1)}{3}= 2\bigl((6-4)+(4-1)\bigr) = 2(d_4+d_3),
\end{align*}
which confirms that $\alpha_i$ represents the contribution of $d_i$ to $\soc{\mathcal{C}}$.  
As observed, $d_4 = 6 - 4$ contributes to the cost of both coalitions $C_1$ and $C_2$, since $\mathcal{C}$ is not a sorted coalition structure. \hfill$\vartriangleleft$
\end{example}
Given a set of agents $[n]$, we define $\mathcal{C}^L=\{C^L_1, C^L_2\}$ to be a \emph{sorted} coalition structure where $C^L_1$ consists of the first $m$ agents and $C^L_2$ contains the remaining ones.
Similarly, we denote $\mathcal{C}^R=\{C^R_1, C^R_2\}$ to represent the \emph{sorted} coalition structure in which $C^R_1$ consists of the last $m$ agents and $C^R_2$ the rest.

\begin{restatable}{lemma}{lemmaMinAlphaAllBefore}
\label{lemma:opt-avg-jump-min-alpha-before-med}
    For a set of agents $[n]$, the coalition structure $\mathcal{C}^L=\{C^L_1, C^L_2\}$ minimizes $\alpha_i$ for all $i \in \{1, \ldots, \nh -1\}$, among all structures with coalitions of sizes $m$ and $n-m$.
\end{restatable}
\begin{proof}
    To prove this, recall from \Cref{eq:alpha-def} that, for a fixed agent~$i$, the value of $\alpha_i(\delta_i)$ depends only on $\delta_i$.  
    Thus, for a given $i$, we can compute the derivative of $\alpha_i(\delta_i)$ with respect to $\delta_i$:
    \begin{align}
        \label{eq:coeff-deriv}
        \frac{d}{d\delta_i} \alpha_i(\delta_i) = \frac{2i-2\delta_i-1}{n-m-1}+\frac{1-2\delta_i}{m-1}.
    \end{align}
    We analyze this derivative separately for the cases $1 \le i \le m$ and $m < i \le \nh-1$.  
    For the first case, $1 \le i \le m$, we have:
    \begin{equation*}
        \frac{d}{d\delta_i} \alpha_i(\delta_i) =
        \begin{cases}
          +, & \delta_i=0, \\
          -, &  \delta_i=i, \\
        \end{cases}
    \end{equation*}
    and for the second case, $m < i \leq \nh-1$, it holds that:
    \begin{equation*}
        \frac{d}{d\delta_i} \alpha_i(\delta_i) =
        \begin{cases}
          +, & \delta_i=0, \\
          -, &\delta_i=m. 
        \end{cases}
    \end{equation*}
    So for a fixed agent $i$ in the first half (i.e., $i \in \{1, \ldots, \nh-1\}$), the minimum of $\alpha_i(\delta_i)$ occurs either at $\delta_i = 0$ or at $\delta_i = \min(i,m)$.
    If $1 \le i \le m$, we have $\delta_i = \min(i,m) = i$, and thus
    \[
        \alpha_i(0) = 2\left(i + \frac{i(1-i)}{n-m-1}\right) , \;\; \alpha_i(i) = 2\left(i + \frac{i(1-i)}{m-1}\right).
    \]
    Since $m \le \nh-1$ and $i(1-i) \le 0$, it follows that
    \begin{align}
        \label{eq:first-half-zero-i}
         \frac{i(1-i)}{n-m-1} \geq \frac{i(1-i)}{m-1} \Rightarrow \alpha_i(0) \geq \alpha_i(i).
    \end{align}
    Hence, for $1 \le i \le m$, the minimum of $\alpha_i(\delta_i)$ occurs at $\delta_i = i$. Thus, for achieving the minimum cost, $\delta_i$ has to grow as fast as possible and thus the $m$ agents in coalition $C_1$ must be the agents having the $m$ smallest values. \\
    On the other hand, for $m < i \le \nh-1$, we have $\delta_i = \min(i,m) = m$, so:
    \begin{align*}
        \alpha_i(m) &= 2\left(i - m + \frac{(i-m)(1-i+m)}{n-m-1}\right) 
        = 2\left(\frac{(i-m)(n-i)}{n-m-1}\right) \\ \notag
        &= 2\left(\frac{in-i^2-mn+im}{n-m-1}\right). 
    \end{align*}
    Also we can rewrite $\alpha_i(0)$ as:
    \begin{align*}
        \alpha_i(0) = 2\left(i + \frac{i(1-i)}{n-m-1}\right) = 2\left(\frac{i(n-m-i)}{n-m-1}\right) = 2\left(\frac{in-im-i^2}{n-m-1}\right).
    \end{align*}
    As a result, it holds that
    \begin{align}
        \label{eq:first-half-zero-m}
        \alpha_i(0) - \alpha_i(m) &= 2\left(\frac{in-im-i^2}{n-m-1} - \frac{in-i^2-mn+im}{n-m-1}\right) \\ \notag
        &= 2\left(\frac{nm-2im}{n-m-1}\right) = 2\left(\frac{m(n-2i)}{n-m-1}\right) \geq 0, 
    \end{align}
    since $i\leq \nh-1$.
    
    Combining (\ref{eq:first-half-zero-i}) and (\ref{eq:first-half-zero-m}), we conclude that $\alpha_i(\delta_i)$ is minimized at $\delta_i = \min(i,m)$ for all $1 \le i \le \nh-1$.  
    This implies that assigning the first $m$ agents to $C_1$ and the remaining agents to $C_2$ yields a coalition structure in which $\alpha_i(\delta_i)$ is minimized for every $1 \le i \le \nh-1$.
\end{proof}

\begin{restatable}{lemma}{lemmaMinAlphaAfterMed}\label{lemma:opt-avg-jump-min-alpha-after-med}
    For a set of agents $[n]$, the coalition structure $\mathcal{C}^R=\{C^R_1, C^R_2\}$ minimizes $\alpha_i$ for all $i \in \{\nhc+1, \ldots, n-1\}$, among all structures with coalitions of sizes $m$ and $n-m$.
\end{restatable}
\begin{proof}
    This statement can be proven using the same approach as in \Cref{lemma:opt-avg-jump-min-alpha-before-med}.
\end{proof}

\begin{restatable}{corollary}{corrAllBeforeMed}\label{corr:opt-avg-jump-before-med}
    For a given coalition structure $\mathcal{C}=\{C_1, C_2\}$ with $|C_1|=m$ and $|C_2|= n - m$, if all the agents in coalition $C_1$ are in $\bigl\{1, 2, \ldots, \nh\bigr\}$, then the coalition structure $\mathcal{C}^L$ has a smaller or equal cost compared with coalition structure $\mathcal{C}$.
\end{restatable}
\begin{proof}
    First, note that by definition, for every $i$, the contribution of $d_i$ to the total cost depends solely on $i$ and $\delta_i$.  
    Thus, as long as only the positions of agents in $C_1$ before or after agent $i$ (including $i$ itself) are changed, $\alpha_i$ remains unchanged. \\
    Since all agents in $C_1$ are chosen from 
    $\bigl\{1, 2, \ldots, \nh \bigr\}$, for every $i > \nh-1$ we have $\delta_i = m$.  
    Hence, selecting other agents from the first half for $C_1$ does not affect $\alpha_i(\delta_i)$ for $i > \nh-1$.\\    
    Consequently, if we choose the $m$ agents from the first half in a way that minimizes $\alpha_i(\delta_i)$ for $i \in \bigl\{1, \ldots, \nh-1 \bigr\}$, we achieve a smaller total cost.\\    
    Applying the result of~\Cref{lemma:opt-avg-jump-min-alpha-before-med}, we conclude that in the coalition structure 
    $\mathcal{C}^L = \{C^L_1, C^L_2\}$, where the first $m$ agents are assigned to $C^L_1$ and the remaining agents to $C^L_2$, all $\alpha_i$ are minimized for $i \in \{1, \ldots, \nh-1\}$.
\end{proof}

\begin{restatable}{corollary}{corrAllAfterMed}\label{lemma:opt-avg-jump-after-med}
    For a given coalition structure $\mathcal{C}=\{C_1, C_2\}$ with $|C_1|=m$ and $|C_2|= n - m$, if all the agents in coalition $C_1$ are in $\bigl\{\nhc+1, \ldots, n\bigr\}$, then the coalition structure $\mathcal{C}^R$, has a smaller total cost.
\end{restatable}
\begin{proof}
    This can be proven using the same approach as \Cref{corr:opt-avg-jump-before-med}, and using \Cref{lemma:opt-avg-jump-min-alpha-after-med}.
\end{proof}

Before stating the next lemma, we introduce the notion of an \emph{isolated agent} in a coalition. We say that an agent $i$ is isolated in coalition~$C$, if none of the agents $i-1$ or $i+1$ are in the coalition~$C$.
Note that we do not consider $v_1\in C_1$ or $v_n\in C_1$ as isolated agents, if $v_2\in C_2$ resp. $v_{n-1}\in C_2$, as for them, the left resp. right neighboring agent is not well-defined.
\begin{restatable}{lemma}{lemmaNoIsolated}\label{lemma:opt-avg-jump-no-isolated}
    An optimal coalition structure $\mathcal{C}=\{C_1, C_2\}$ cannot have an isolated agent.
\end{restatable}

\begin{proof}
    Let $|C_1| = m$ and $|C_2| = n-m$, and, without loss of generality, let $j$ be an isolated agent in coalition $C_1$, so that agents $j-1$ and $j+1$ belong to $C_2$.  
    The case where $j$ is instead in $C_2$ can be proven analogously.
    Moreover, let $\cev{\mathcal{C}}$ denote the coalition structure obtained by assigning agent $j-1$ to $C_1$ and agent $j$ to $C_2$; that is, by swapping agents $j$ and $j-1$.  
    Note that in $\cev{\mathcal{C}}$, only $\delta_{j-1}$ increases by~1, while all other $\delta_i$ for agents $i \neq j$ remain unchanged.  
    Hence, we have $\cev{\delta}_{j-1} = \delta_{j-1} + 1 = \delta_j$, since in $\mathcal{C}$, we have $\delta_{j-1} = \delta_j - 1$.  
    Let $\cev{\alpha}_{j-1}$ denote the contribution of $d_{j-1}$ to $\soc{\cev{\mathcal{C}}}$. Then:
    \begin{align*}
    \cev{\alpha}_{j-1}
      &= \alpha_{j-1}(\delta_{j-1}+1) = \alpha_{j-1}(\delta_{j})\\
      & = 2\biggl( j-1+\frac{(1-\delta_{j})(\delta_{j})}{m-1}+\frac{(j-1-\delta_{j})(2-j+\delta_{j})}{n-m-1}\biggr),
    \end{align*}
    and its difference to $\alpha_{j-1}=\alpha_{j-1}(\delta_{j-1})$, using \Cref{eq:alpha-def}, will~be:
    \begin{align}
        \label{eq:diff-left-one}
        \cev{\alpha}_{j-1} - \alpha_{j-1} &= 2\biggl(\frac{2(1-\delta_j)}{m-1}+\frac{2(j-\delta_j-1)}{n-m-1}\biggr).
    \end{align}
    Conversely, consider the coalition structure $\vec{\mathcal{C}}$ obtained by assigning agent $j+1$ to $C_1$ and agent $j$ to $C_2$, that is, by swapping agents $j$ and $j+1$.  
    In $\vec{\mathcal{C}}$, only $\delta_j$ decreases by one, while all other $\delta_i$ for $i \neq j$ remain unchanged.  
    Let $\vec{\alpha}_{j}$, denote the contribution of $d_{j}$ to $\soc{\vec{\mathcal{C}}}$.
    Then:
    \begin{align*}
    \vec{\alpha}_j
      &= \alpha_j(\delta_{j}-1)\\
      & = 2\biggl( j+\frac{(\delta_{j}-1)(2-\delta_{j})}{m-1} + \frac{(j-\delta_{j}+1)(\delta_{j}-j)}{n-m-1}\biggr)
    \end{align*}
    and its difference to $\alpha_{j}$, using \Cref{eq:alpha-def}, will be:
    \begin{align}
        \label{eq:diff-right-one}
        \vec{\alpha}_{j} - \alpha_j &= 2\biggl(\frac{2(\delta_j-1)}{m-1}+\frac{2(\delta_j-j)}{n-m-1}\biggr).
    \end{align}
    Combining \Cref{eq:diff-left-one} and \Cref{eq:diff-right-one}, we get:
    \begin{align*}
        \bigl(\cev{\alpha}_{j-1} - \alpha_{j-1}\bigr)+\bigl(\vec{\alpha}_{j} - \alpha_j\bigr) &= \frac{-4}{n-m-1}.
    \end{align*}
    This means that at least one of $\bigl(\cev{\alpha}_{j-1} - \alpha_{j-1}\bigr)$ or $\bigl(\vec{\alpha}_j - \alpha_j\bigr)$ is less than or equal to zero.  
    Consequently, at least one of the coalition structures, $\cev{\mathcal{C}}$ or $\vec{\mathcal{C}}$, results in a lower cost than $\mathcal{C}$.
\end{proof}

\begin{restatable}{lemma}{LemmaMiddleBlob}
    \label{lemma:opt-avg-jump-middle-blob}
    For a coalition structure $\mathcal{C}=\{C_1, C_2\}$, where $C_1$ consists of $m$ consecutive agents, indexed from $\ell$ to $r$, with ${\ell \leq \lfloor{n/2}\rfloor < r}$, at least one of the coalition structures $\mathcal{C}^L$ or $\mathcal{C}^R$ has a smaller cost.
\end{restatable}
\begin{proof}
        Let $\Delta = \lfloor{n/2}\rfloor - \ell$ and w.l.o.g. assume that $\Delta \leq m/2$. In this case, we show that $C^R$ has smaller social cost compared to $C$.
        
        For that, let $\alpha^R_i$ and $\delta^R_i$ denote the contribution of $d_i$ to the cost of $\mathcal{C}^R$ and the number of agents before agent $i$ (including itself) in coalition $C^R_1$, respectively.
        By~\Cref{eq:alpha-def}, we get that the value of $\alpha_i$ for $1\leq i \leq \ell-1$ is the same in both $\mathcal{C}$ and $\mathcal{C}^R$, as $\delta^R_i=\delta_i=0$ for all these agents in both coalition structures.
        On the other hand, by \Cref{lemma:opt-avg-jump-min-alpha-after-med}, we know that in $\mathcal{C}^R$, all distances in the second half, that is $\nh <i\leq n-1$, have the smallest possible $\alpha_i$.
        So it remains to show that $\alpha^R_i$ is not bigger than $\alpha_i$, for  $\ell \leq i \leq \nh$.
        For these distances in $\mathcal{C}$, using \Cref{eq:alpha-def}, we have:
        \begin{align*}
            \alpha_i(\delta_i) &= 2\left(i+\frac{\delta_i(1-\delta_i)}{m-1}+\frac{(i-\delta_i)(1-i+\delta_i)}{n-m-1}\right) \\
            & = 2\left(i + \frac{\delta_i(1-\delta_i)}{m-1} + \frac{(\ell - 1)(2-\ell)}{n-m-1}\right),
        \end{align*}
        and in $\mathcal{C}^R$, we have:
        \begin{align*}
            \alpha^R_i(\delta_i) &= 2\left(i+ \frac{i(1-i)}{n-m-1}\right).
        \end{align*}
        Computing their difference:
        \begin{align*}
            &\alpha_i(\delta_i) - \alpha^R_i(\delta_i) \\
            =\;& 2\left(\frac{\delta_i(1-\delta_i)}{m-1} + \frac{(\ell - 1)(2-\ell)-i(1-i)}{n-m-1}\right)\\
            =\;& 2\left(\frac{\delta_i(1-\delta_i)}{m-1} + \frac{(\ell - 1)(2-\ell)+i(i-1)}{n-m-1}\right).
        \end{align*}
        Note that since $i \geq \ell-1 +\delta_i$, it holds that 
        \begin{align}
            &i(i-1)\notag\\
            \geq\;& (\ell-1 +\delta_i)(\ell-1 +\delta_i-1) \notag \\ 
            =\;& (\ell-1)(\ell +\delta_i-2)+\delta_i(\ell +\delta_i-2).\label{eq:multi-bound}
        \end{align}
        Applying \Cref{eq:multi-bound} to the difference of $\alpha_i(\delta_i)$ and $ \alpha^R_i(\delta_i)$, we get:
        \begin{align*}
            &\alpha_i(\delta_i) - \alpha^R_i(\delta_i)\\
            \geq\;& 2\bigg(\frac{\delta_i(1-\delta_i)}{m-1} + \frac{\delta_i(\ell +\delta_i-2)}{n-m-1} + \frac{(\ell - 1)(2-\ell+\ell +\delta_i-2)}{n-m-1}\bigg) \\
            =\;& 2\bigg(\frac{\delta_i(1-\delta_i)}{m-1} + \frac{\delta_i(\ell +\delta_i-2)}{n-m-1} + \frac{\delta_i(\ell - 1)}{n-m-1}\bigg) \\
            =\;& 2\delta_i \left( \frac{2\ell - 3 +\delta_i}{n-m-1}+\frac{1-\delta_i}{m-1} \right).
        \end{align*}
        Further, it holds that $2(\ell-1) \geq n - m$ as $\ell-1$ is the size of the left part of coalition $C_2$ and $n-m$ is the full size of $C_2$. Thus, it follows
        
        \begin{align*}
            \alpha_i(\delta_i) - \alpha^R_i(\delta_i) &\geq  2\delta_i \left(\frac{n-m-1}{n-m-1}+\frac{\delta_i}{n-m-1}+\frac{1-\delta_i}{m-1} \right) \\
            &= 2\delta_i \left(1+\frac{\delta_i}{n-m-1}+\frac{1-\delta_i}{m-1} \right) \\
            &= 2\delta_i \left(\frac{\delta_i}{n-m-1}+\frac{m-\delta_i}{m-1} \right) > 0.
        \end{align*}
        This is greater or equal to zero, since all variables are non-negative and $\delta_i \leq m$. Additionally, since we only consider indices starting at $\ell$, we know that $\delta_i$ is at least $1$. Thus, this term is even strictly greater than zero.
        As a result, we observe that $\alpha^R_i$ is smaller than $\alpha_i$ for $\ell \leq i \leq \nh$.
        Consequently, $\mathcal{C}^R$ has a smaller cost than $\mathcal{C}$.\\
        For the other case, when $\Delta> m/2$, by the same approach, we can prove that $\mathcal{C}^L$ has a smaller cost.
\end{proof}

Based on the above lemmata and observations, we propose the following conjecture.
\avJumpOptConj*

This conjecture indicates that \textsc{Average-Jump} games may follow a fundamental structural rule. If true, the optimal coalition structure would be sorted for any $k$, as otherwise any two overlapping coalitions $C_i$ and $C_j$ in $\mathcal{C}$ could be reordered to lower the overall cost.

\section{Omitted Details from Section 4.1}
\label{appx:sub:PoA}
\lemmaPoACutoff*
\begin{proof}[Full Proof of \Cref{lemma:PoA-cutoff}]
    We give two different game instances for $k=2$ and $k > 2$ respectively, and show that they have an optimum with social cost $0$ and admit a coalition structure with positive social cost that is both swap and jump stable. For all constructions, we assume a specific~$\lambda$. However, they can be scaled to work for all~$\lambda$.

    For $k=2$, assume $\lambda=1$ and consider the game with $n\coloneqq5$ agents with values $(v_i)_{i \in [5]} \coloneqq (1, 1-\varepsilon, 1, 1+\varepsilon, 2)$ for some $\varepsilon \in (0,\frac{1}{2}\lambda)$. For the \textsc{Swap} instance additionally assume coalition sizes of $k_1\coloneqq2$ and $k_2\coloneqq3$. For both, the \textsc{Swap}- and the \textsc{Jump} instance, the coalition structure
        $$\mathcal{C}^* \coloneqq \big\{\{0,1-\varepsilon,1\},\{1+\varepsilon,2\}\big\}$$
    is the social optimum as it has a social cost of $0$. However, consider the coalition structure
        $$\mathcal{C} \coloneqq \big\{\{0,1,2\},\{1-\varepsilon,1+\varepsilon\}\big\}$$
    which has a social cost of $1$. Note that the coalition structure $\mathcal{C}$ is swap stable because there is one coalition with only friends. Further, it is jump stable, because the only agents with a positive cost have exactly one friend and one enemy in both coalitions and can thus not improve. This shows an unbounded PoA of \textsc{Cutoff$_\lambda$-\{Swap, Jump\}} games with $k=2$ coalitions.

    For $k>2$, consider the game with $n\coloneqq 4(k-1)$ agents and coalitions of size $k_1 \coloneqq 2(k-1)$ and $k_i = 2$ for all $i \in [k] \setminus \{1\}$. The agents have values to fit exactly into the coalitions, i.e., there are $2(k-1)$ agents with value $0$ and for all $\{2,\dots,k\}$ exactly two agents with that value each. In this setting,
        $$\mathcal{C}^* \coloneqq \big\{\{\overbrace{0\ldots0}^{2(k-1)}\},\,\{2,2\},\,\dots,\{k,k\}\big\}$$
    is the social optimum as it has a social cost of $0$. However, placing all $0$-valued agents in the coalitions of size $2$ forces the other agents to be in the large coalition of size $2(k-1)$. This results in the coalition structure
        $$\mathcal{C} \coloneqq \big\{\overbrace{\{0,0\},\dots,\{0,0\}}^{k-1},\,\{2,2,\,3,3,\,\,\dots,\,k,k\}\big\}.$$
    For all $k>2$ there are at least two enemies in the coalition last coalition of $\mathcal{C}$, yielding a strictly positive social cost. Again, this coalition structure is indeed a swap equilibrium because it is sorted (see \Cref{theorem:swap-sorted-eq-exists}).\footnote{Note that this construction can be extended to work for all even $n \geq 4(k-1)$. Then, the coalition sizes are defined to satisfy $k_i \geq 2$ for all $i \in [k]\setminus\{1\}$ and $k_1 = \frac{n}{2}$, and the values fit exactly into these coalitions, i.e., there are $k_i$ many agents with value $i$ for all $i \in \{2,\dots,k\}$ and $k_1 = \frac{n}{2}$ many agents with value $0$.}

    Further, the coalition structure $\mathcal{C}$ is also jump stable, since all agents with value $0$ have minimum cost and all other agents are in the only coalition with their single friend.
\end{proof}

\lemmaPoAMaxAvgJumpSwap*
\begin{proof}[Full Proof of \Cref{lemma:PoA-Max-Avg-Jump-Swap}]
    We will first explain a minimal instance with $k=3$, $n=8$, and then explain how this can be adjusted to work for larger $n$ and $k$. Let $L<M<H$ and $M = \frac{H+L}{2}+1$. The following is an equilibrium with asymptotically maximal cost of $\bigO((H-L)n)$:
        $$\big\{\{L,L\},\{L,L\},\{M,M,H,H\}\big\}.$$
    This is swap stable because of \Cref{theorem:swap-sorted-eq-exists}. It is jump stable, as all but the last coalition have minimal cost, and a jump of any single agent from the last coalition to another would increase their cost.
    All coalitions contain at least two agents, so there is no difference for \textsc{UIS} or \textsc{HIS}. All coalitions except for the last one have cost 0. The last has cost in $\Theta\left((H-L)n\right)$.
    The corresponding social optimum with cost 0 is structured as follows:
        $$\big\{\{L,L,L,L\},\{M,M\},\{H,H\}\big\}.$$
    Again, all coalitions have at least two agents, so there is no difference for \textsc{UIS} or \textsc{HIS}. 
    The instance can be generalized in the following~way:
        $$\big\{\overbrace{\{L,L\}}^{k-3},\{L,L\},\{L,L\},\{\overbrace{M\ldots M}^{\frac{(n-2k-2)}{2}},\overbrace{H \ldots H}^{\frac{(n-2k-2)}{2}}\}\big\},$$
    by adding pairs of $L$-agents for larger $k$, and a balanced amount of $M$'s and $H$'s for larger $n\geq 3k+2$.
    
    For all \textsc{Max} games, consider the following variant of the construction of \Cref{lemma:max-jump-unsorted-eq} with only two coalitions:
    $$\big\{\{L,L,L,L\},\{H,H,H,H\}\big\}$$ 
    having optimal social cost 0, and
    $$\big\{\{L,L,H,H\},\{L,L,H,H\}\big\}$$
    having two coalitions with maximal cost in $\Theta(n~(H~-~L))$.
    The second instance is an equilibrium, independently if it is a \textsc{Swap}, \textsc{Jump-UIS}, or \textsc{Jump-HIS} game, as the cost for any agent is the same in both coalitions. So for \textsc{Max} games, PoA is unbounded even for $k=2$.
\end{proof}

\lemmaAvgJumpHISkTwoPoA*
\begin{proof}[Proof of \Cref{lemma:avg-jump-his-k2-poa}]
    Let us first look at the $\SC(\OPT)=0$ case. We only consider the case with $n\geq 4$ agents and each group has at least two agents, as otherwise the unique optimum becomes the unique equilibrium because agents would simply move into empty coalitions or to their own type. 
    We will show that the optima are the only equilibria.
    Since we have only $k=2$ coalitions, either (1) all agents share the same value, or (2) there are two different agent values $L$ and $H$.
    For (1), no agent in equilibrium pays any cost, so any equilibrium also has cost 0.
    For (2), we have $x$ agents with value $L$, and $n-x$ agents with value $H$.
    In the social optimum, similar to the $k=2$ example in \Cref{lemma:PoA-Max-Avg-Jump-Swap}, there is one coalition with only $L$ agents, and one with only $H$ agents.
    Assume for contradiction there is an equilibrium with a different distribution, where there is at least one coalition with agents of both values. In the \textsc{HIS} setting, the grand coalition would not be stable, as all agents have cost larger than 0 and moving to the empty coalition would decrease their cost.
    If there is only one coalition with 2 different agent values, wlog $C_1$, the agents in the other coalition $C_2$ have the same value, wlog $H$, and thus a cost of $0$. But then, all $H$-agents in coalition $C_1$ can improve by jumping to coalition $C_2$.
    If, however, both coalitions have agents of both types, there is one agent with value $L$ in $C_1$ and $C_2$. We call them $a_1$ and $a_2$, and assume wlog $\cost(a_1,C_1)\leq \cost(a_2,C_2)$. This, however, means that $a_2$ can improve by jumping to $C_1$, as $\cost(a_2,C_1\cup\{a_2\}) = \frac{|C_1|-1}{|C_1|} \cost(a_1,C_1)$.
    Both cases lead to a contradiction, meaning the only stable coalition structure is the structure with one coalition of only $L$ agents, and one with only $H$ agents.

    Before we look at the second case, we can see that any instance with $k=2$ coalitions and at least three different values has $\SC(\OPT)>0$, as per pigeonhole principle, at least two agents with different values need to be in at least one coalition, both having cost larger than 0. This rules out an unbounded PoA for this setting.

    Next, let us consider the case where $\SC(\OPT)>0$. 
    Here, we provide the following instance with PoA in $\Omega(n)$. Let $L=1$ and $H=n$, and an instance be given with one value $L=1$ agent, one $H=n$ agent, and $n-2$ agents of value $L+1=2$. We look at the following two coalition structures:
    $$\big\{\{1,\overbrace{2\ldots2}^{n-2}\},\{n\}\big\},$$
    with social cost $1+\frac{1}{n-2}\cdot(n-2)+1=2$, and
    $$\big\{\{1\},\{\overbrace{2\ldots2}^{n-2},n\}\big\},$$
    with cost $0+\frac{n-2}{n-2}\cdot(n-2)+(n-2)=2(n-2)$. The first is an upper bound to the social optimum, and if the second coalition structure is an equilibrium, the PoA is lower bounded in $\Omega(n)$.
    We can see that the second structure is indeed an equilibrium, as the agent with value $1$ has optimal cost of 0, and the only other possible moves would be (1) agent $n$ moving to the left coalition, which increases its cost from $(n-2)$ to $n-1$, or (2) one agent with value $2$ moving to the left coalition, keeping its cost from $\frac{n-2}{n-2}=1$ to $2-1=1$. In all cases jumping would not be beneficial, and thus, the second structure is an equilibrium.
    By comparing the cost of both coalition structures, we get a cost ratio of $\frac{2(n-2)}{2}=n-2$ meaning the PoA is indeed in $\Omega(n)$.
\end{proof}

\lemmaAvgJumpUISkTwoPoA*
\begin{proof}[Proof of \Cref{lemma:avg-jump-uis-k2-poa}]
    Similar to the construction of \Cref{lemma:max-jump-unsorted-eq}, by comparing $$\big\{\{L,L\},\{H,H\}\big\}$$ having two occupied coalitions and social cost 0, to
    $$\big\{\{\},\{L,L,H,H\}\big\}$$
    having one occupied coalition and cost $\Theta(n(H-L))$ resp. $\Theta(n)$ for \textsc{Cutoff} games with $\lambda <\dist(H,L)$, we can see the PoA is unbounded.
    The first coalition structure is a social optimum, as its social cost is minimal, and the second structure, a grand coalition, is stable for all three \textsc{Jump-UIS} cost settings by \Cref{obs:Jump-UIS-grand-coalition-stable}.
\end{proof}

\section{Omitted Details from Section 4.2}
\label{appx:sub:PoS}
\corCutoffAvgSwapPos*
\begin{proof}
    This follows from the fact that the social cost is a potential function for \textsc{Cutoff$_\lambda$-Swap} and \textsc{Avg-Swap} games (confer to \Cref{lemma:swap-mfhg} and \Cref{cor:swap-potential-games}).

    Assume for contradiction that there is a game instance (with the social cost being a potential function) such that the social optimum is not an equilibrium. Since it is not an equilibrium, there is an improving move, which would then strictly decrease the potential function. As the potential function is the social cost, there would be another coalition structure with a lower social cost than the social optimum, which is a contradiction.
\end{proof}

\lemmaMaxSwapPoS*
\begin{proof}[Proof of \Cref{lemma:max-swap-pos}]
    Assume for contradiction, that there is a coalition structure $\mathcal{C}$ with optimum cost that has an improving swap. Let $x$ and $y$ be two agents with an improving swap, and $\mathcal{C}^\prime$ be the coalition structure $\mathcal{C}$ after the swap of agents $x$ and $y$. We assume wlog that $v_x < v_y$, since agents with the same value would not strictly decrease their cost by swapping. 
    Further, $X \coloneqq \mathcal{C}(x) \setminus \{x\}$ and $Y \coloneqq \mathcal{C}(y) \setminus \{y\}$ the sets of other agents in the coalitions of agent $x$ and agent $y$ respectively, and let $X_{\min}$ ($Y_{\min}$) and $X_{\max}$ ($Y_{\max}$) to be the minimum and maximum value of $X$ (and $Y$). Note that all costs of all agents in the coalitions $\mathcal{C}(x)$ and $\mathcal{C}(y)$ are the distances to these four values. Thus, if $v_x \in [Y_{\min},Y_{\max}]$, the costs of all agents in $Y$ do not change, and similarly, if $v_y \in [X_{\min},X_{\max}]$, the costs of all agents in $X$ do not change. 

    First, assume that both $v_x \in [Y_{\min},Y_{\max}]$ and $v_y \in [X_{\min},X_{\max}]$ hold. In this case, all costs of all agents they the same except for the costs of agents $x$ and $y$ that decrease with the swap by definition. Therefore, the coalition structure $\mathcal{C}^\prime$ would have lower social cost than $\mathcal{C}$ which is a contradiction to $\mathcal{C}$ being a social optimum. 

    As a result, we know that either $v_x \notin [Y_{\min},Y_{\max}]$ or $v_y \notin [X_{\min},X_{\max}]$. We assume wlog that $v_x \notin [Y_{\min},Y_{\max}]$. Since we assume $x < y$, we know that $v_x < Y_{\min} \leq Y_{\max}$. Now, there are two cases depending on whether the cost of agent $y$ are defined by the distance to $Y_{\min}$ or $Y_{\max}$. 

    If $\cost(y,Y) = |y - Y_{\min}|$, we know that the cost for agent $y$ increase by the swap since $X_{\min} \leq x < Y_{\min}$. This is a contradiction to the swap being improving for both agents.

    If, however, the cost of agent $y$ in coalition structure $\mathcal{C}$ are $\cost(y,Y) = |Y_{\max} - y|$, consider the following chain of inequalities:
    \begin{align*}
        |v_y-X_L| \overset{(1)}{>} |v_x-X_L| \overset{(2)}{>} |Y_R-v_x| \overset{(3)}{>} |v_y-Y_R|.
    \end{align*}
    Note that inequalities (1) and (3) simply hold by the assumption that $v_x < v_y$. Additionally, Inequality (2) holds, because otherwise the swap would not be improving for agent $x$. However, in total, this chain of inequality shows, that the swap is then not improving for agent $y$, which is a again a contradiction. 

    As a result, there cannot be a improving swap in a coalition structure with optimal cost, proving a price of stability of $1$ for all \textsc{Max-Swap} games.
\end{proof}

\lemmaAvgMaxJumpPoS*
\begin{proof}[Proof for \textsc{Average}]
    Let us first look at the social optimum. It is unstable, as agent 4 can jump from right to left and decrease its cost from $\frac{2+4+4}{3}=\frac{10}{3}>3$ to $\frac{3+3+3}{3}=3$. We checked that this is indeed a social optimum by using brute force enumeration of all agent assignments. The rationale behind this coalition structure being the optimum, is that all 1's need to stay separate of all 8's, leaving two possible other assignments:
    $$\big\{\{1,1,1,\mathbf{4,6}\},\{8,8\}\big\}$$
    with cost $\frac{3\cdot(3+5)+(3\cdot3+2)+(3\cdot5+2)}{4}=\frac{24+11+17}{4}=13$, and
    $$\big\{\{1,1,1,\mathbf{6}\},\{\mathbf{4},8,8\}\big\}$$
    with cost $\frac{3\cdot5+3\cdot5}{3}+\frac{2\cdot4+2\cdot4}{2}=10+8=18$.
    
    The second structure is a jump equilibrium, as no agent can improve by jumping to the other coalition. Independently of which agent is considered, we can see that in its current coalition, the maximum distance to any other agent is at most as large as the minimum distance to any agent in the other coalition, thus its current cost is less or equal to its cost in the other coalition. Thus, no jump would be improving.  
\end{proof}
\begin{proof}[Proof for \textsc{Maximum}]
    The coalition structure $\mathcal{C}^*$ is unstable as the agent with value $4$ can decrease its cost from $4$ to $3$ by jumping to the other coalition. 

    In the remainder of this proof, we show that $\mathcal{C}^*$ is indeed the unique social optimum (with social cost of $14$), implying that the PoS greater than $1$. First, we show that all coalition structure in which agents with values $1$ and $8$ are in the same coalitions have a higher cost than $14$. In this case, these two agents alone contribute at least a cost of $14$ to the social cost. Then, either there is another agent in their coalition (causing an additional cost of at least $3$), or the agents with values $4$ and $6$ are in the other coalition (causing an additional cost of at least $4$), all of these coalition structures have a higher cost than $\mathcal{C}^*$. For the social cost of the remaining coalition structures, consider the following list:
    \begin{itemize}
        \item $\mathcal{C}^* =\{\{1,1,1\},\{4,6,8,8\}\}$ has social cost $14$
        \item $\{\{1,1,1,4\},\{6,8,8\}\}$ has social cost $18$
        \item $\{\{1,1,1,6\},\{4,8,8\}\}$ has social cost $32$
        \item $\{\{1,1,1,4,6\},\{8,8\}\}$ has social cost $23$
    \end{itemize}
    As a result, $\mathcal{C}^*$ is the unique optimum that is also unstable, showing that the PoS for \textsc{Max-Jump} games is greater than $1$.
\end{proof}

\lemmaPoSCutoffJump*
\begin{proof}
    For the first part, let $\mathcal{I}\coloneqq(n,(v_i)_{i\in [n]},k)$ be a nice \textsc{Cutoff$_\lambda$-Jump} game instance and $\mathcal{B}$ a set of $\lambda$-blocks, that cover the values $V\coloneqq\{v_i\}_{i \in [n]}$ of the game instance $\mathcal{I}$. We assume wlog that all $\lambda$-blocks are disjoint. We show that the coalition structure of $\mathcal{I}$ with the minimum social cost is also a jump equilibrium. For that, note that all agents whose values are part of the same $\lambda$-block are pairwise friends. Therefore, if $V[B]$ is the set of agents whose values are covered by a $\lambda$-block $B \in \mathcal{B}$, then $\mathcal{C} \coloneqq \{V[B]\}_{B \in \mathcal{B}}$ is a coalition structure with a social cost of $\SC(\mathcal{C}) = 0$. This means, that every agent has a cost of $0$ in this coalition structure, implying that $\mathcal{C}$ is both optimum and jump stable. This shows that the PoS is exactly $1$ for nice \textsc{Cutoff$_\lambda$-Jump} games.

    For the second part, we give a \textsc{Cutoff$_\lambda$-Jump} game instance that is not nice, where the social optimum is not jump stable. For that consider the game instance with $n=8$ agents, $k=2$ coalitions and values 
        $$\left\{0, \frac14\varepsilon, \frac24\varepsilon, \frac24\varepsilon, \lambda + \frac14\varepsilon, \lambda + \frac34\varepsilon, 2\lambda + \frac14\varepsilon, 2\lambda+\varepsilon\right\}.$$
    for every $\varepsilon > 0$. Note that the values of this instance can be covered by an interval of size $2\lambda + \varepsilon$, which to just by $\varepsilon$ too large for the instance to be nice. The social optimum of this game instance is
        $$\mathcal{C}^* \coloneqq \left\{\left\{0, \frac14\varepsilon, \frac24\varepsilon, \frac24\varepsilon\right\},\; \left\{ \underline{\lambda + \frac14\varepsilon}, \lambda + \frac34\varepsilon, 2\lambda + \frac14\varepsilon, 2\lambda+\varepsilon\right\}\right\}$$
    with a social cost of $\frac{4}{3}$. However, the agents with value $\lambda + \frac14\varepsilon$ can jump to the other coalition to decrease its cost from $\frac{1}{3}$ to $\frac{1}{4}$.

    In the rest of this proof, we show that $\mathcal{C}^*$ is indeed the optimum coalition structure. To avoid an exhausting case distinction, we show that many sets of coalition structures are not better than $\mathcal{C}^*$. 
    First, note that in all coalition structures that are at least as good as $\mathcal{C}^*$, there may not be any agent with a cost of $1$, as this agent would induce a cost $1$ on its own and an additional cost of $|C|-1\cdot \frac{1}{|C|-1} = 1$ because the other agent has at least one enemies in its coalition. Therefore, such coalition structures have a cost of at $2$, which is greater than the cost of $\frac{4}{3}$ of the coalition structure $\mathcal{C}^*$. Since the agent with value $v_8 = 2\lambda+\varepsilon$ only has the agent with value $v_7 = 2\lambda + \frac14\varepsilon$ has friend, they have to be in the same coalition. 
    
    Next, we rule out that coalition structures with coalitions of size $1$ and $2$. For that, we consider the number of enemies each agent has in the whole game instance (see \Cref{table:pos-cutoff-jump}). Note that the sum of these numbers is $30$. 

    \begin{table}[!h]
    \begin{center}
        \begin{tabular}{r|c|c|c|c|c|c|c|c}
            $i$ & $1$ & $2$ & $3$ & $4$ & $5$ & $6$ & $7$ & $8$ \\ \hline
            value & $0$ & $\frac{\varepsilon}{4}$ & $\frac24\varepsilon$ & $\frac24\varepsilon$ & $\lambda{+}\frac{\varepsilon}{4}$ & $\lambda{+}\frac34\varepsilon$ & $2\lambda{+}\frac{\varepsilon}{4}$ & $2\lambda{+}\varepsilon$ \\ \hline
            $|N_i^-(V)|$ & $4$ & $3$ & $3$ & $3$ & $2$ & $5$ & $4$ & $6$
        \end{tabular}
    \end{center}
    \caption{Number of enemies over all agents in the \textsc{Cutoff$_\lambda$-Jump} game instance used in part (ii) of \Cref{lemma:pos-cutoff-jump}.\label{table:pos-cutoff-jump}}
    \end{table}
    
    If we assume a coalition structure with coalitions of size $1$ and $7$, we know that the cost of the agent in the singleton coalition is at least $0$. For the coalition of size $7$, we know that the sum of enemies in this coalition over all agents is at least $17$ ($30$ in total, minus $6$ for putting the agent with the most enemies in the singleton coalition, and minus $7$ even if the agent in the singleton coalition was an enemies of all other agents). Thus, the total cost of such a coalition structure is at least $\frac{17}{|7|-1} > \frac{4}{3} = \SC(\mathcal{C}^*)$. With the same arguments we can rule out any coalition structure with a coalitions of size $2$ and $6$ that have a cost of least $\frac{30-11-2\cdot 6}{|6|-1} > \frac{4}{3}$. As a result, we know that in the optimum coalition structure both coalitions have to have a size of at most $5$. 

    We already know that the agents $7$ and $8$ have to be in the same coalition. In the following, we make a case distinction over the size of the coalition they are part of. If they are part of a coalition of size $5$, then agent $8$ has at least three enemies in this coalition and thus has a cost of $\frac{3}{4}$ and induces a cost of at least $\frac{3}{4}$ to the other agents. As these cost alone are higher than $\frac{4}{3} = \SC(\mathcal{C}^*)$, we know that agents $7$ and $8$ cannot be part of a coalition of size $5$. If agents $7$ and $8$ are part of a coalition of size $3$, then only agents $5$ and $6$ could be the remaining agent in this coalition, since all other agents had a cost of $1$ in this coalition. These two coalition structures have a social cost of
    \begin{itemize}
        \item $3$ if agent $5$ joins agents $7$ and $8$, and
        \item $1.5$ if agent $6$ joins agents $7$ and $8$,
    \end{itemize}
    which is both more than $\frac{4}{3} = \SC(\mathcal{C}^*)$.

    If agents $7$ and $8$ are part of a coalition of size $4$, there are only a few coalition structures left to check under the above findings (and up to symmetry). These coalition structures have a social cost of 
    \begin{itemize}
        \item $4.0$ for $\big\{\{0, \frac24\varepsilon, 2\lambda{+}\frac14\varepsilon, 2\lambda+\varepsilon\},\; \{ \frac14\varepsilon, \frac24\varepsilon, \lambda{+}\frac14\varepsilon, \lambda{+}\frac34\varepsilon \}\big\}$
        \item $4.0$ for $\big\{\{0, \frac24\varepsilon, \lambda{+}\frac14\varepsilon, 2\lambda+\varepsilon\},\; \{ \frac14\varepsilon, \frac24\varepsilon, \lambda{+}\frac34\varepsilon, 2\lambda{+}\frac14\varepsilon\}\big\}$
        \item $\frac{14}{3}$ for $\big\{\{0, \frac24\varepsilon, \lambda{+}\frac14\varepsilon, \lambda{+}\frac34\varepsilon,\},\; \{ \frac14\varepsilon, \frac24\varepsilon, 2\lambda{+}\frac14\varepsilon, 2\lambda+\varepsilon\}\big\}$
        \item $\frac{14}{3}$ for $\big\{\{0, \frac14\varepsilon, \lambda{+}\frac14\varepsilon, \lambda{+}\frac34\varepsilon,\},\; \{ \frac24\varepsilon, \frac24\varepsilon, 2\lambda{+}\frac14\varepsilon, 2\lambda+\varepsilon\}\big\}$
        \item $\frac{4}{3}$ for $\mathcal{C}^*$
    \end{itemize}
    showing that $\mathcal{C}^*$ is the social optimum.
\end{proof}

\end{document}